\documentclass[conference]{IEEEtran}

\usepackage{amsmath,amssymb,amsthm,mathtools}
\usepackage{booktabs}
\usepackage{array}
\usepackage{enumitem}
\usepackage{xspace}
\usepackage{microtype}
\usepackage{url}
\usepackage[hidelinks]{hyperref}
\usepackage[nameinlink,noabbrev]{cleveref}
\usepackage{tikz}
\usetikzlibrary{arrows.meta,positioning,fit,backgrounds}

\newif\ifarxiv
\arxivfalse
\arxivtrue

\newtheorem{definition}{Definition}
\newtheorem{theorem}{Theorem}
\newtheorem{lemma}{Lemma}
\newtheorem{proposition}{Proposition}
\newtheorem{corollary}{Corollary}

\newcommand{\Fresh}{\ensuremath{\mathsf{new}}\xspace}
\newcommand{\Alias}{\ensuremath{\mathsf{reuse}}\xspace}
\newcommand{\Use}{\ensuremath{\mathsf{Use}}\xspace}
\newcommand{\Dispatch}{\ensuremath{\mathsf{Dispatch}}\xspace}
\newcommand{\Retry}{\ensuremath{\mathsf{Retry}}\xspace}
\newcommand{\Resolve}{\ensuremath{\mathsf{Resolve}}\xspace}

\newcommand{\Lin}{\ensuremath{\mathsf{Lin}}\xspace}
\newcommand{\Partial}{\ensuremath{\mathord{\downarrow}}\xspace}
\newcommand{\sem}[1]{\mathopen{[\![}#1\mathclose{]\!]}}

\title{When Can Agents Safely Checkpoint, Fork, Restore, and Merge?\\
Exact Checking for Execution Edits}

\ifarxiv
\author{\IEEEauthorblockN{Yusheng Zheng\IEEEauthorrefmark{1},
Xiaoyu Song\IEEEauthorrefmark{2},
Yanpeng Hu\IEEEauthorrefmark{3},
Lebin Cheng\IEEEauthorrefmark{4},
Yuxi Huang\IEEEauthorrefmark{5}, and
Wei Zhang\IEEEauthorrefmark{6}}
\IEEEauthorblockA{\IEEEauthorrefmark{1}University of California, Santa Cruz\\
\IEEEauthorrefmark{2}University of Texas at Dallas\\
\IEEEauthorrefmark{3}ShanghaiTech University\\
\IEEEauthorrefmark{4}RailXia \quad
\IEEEauthorrefmark{5}Eunomia Labs \quad
\IEEEauthorrefmark{6}University of Connecticut}}
\else
\author{}
\fi

\begin{document}
\maketitle

\begin{abstract}
Agent runtimes can Checkpoint an execution, Fork it, Restore a checkpoint, or Merge branches without restarting a task.
We call these operations \emph{execution edits}, with Checkpoint recording the current execution for later use and Fork, Restore, and Merge changing what the Agent will do next.
An execution edit cannot undo an earlier authorization or a tool request already sent.
An unsafe edit can therefore authorize the same tool action twice, discard a result the task still requires, or conflict with a call that began before the edit.
The Agent is untrusted, so the runtime uses its execution record to determine which past actions an edit must account for and which required results it must preserve to keep the subsequent execution safe.
Yet existing Agent systems support such operations without deriving what each edit must preserve from the running execution, whereas prior methods for computing safe behavior take that requirement as input.
We give an algorithm that decides exactly whether an edit is safe.
It returns all safe ways to continue, or proves that none exists.
To make this decision, the algorithm lists every way the task can finish without violating policy.
It removes any way that could make a still-required result impossible to finish later.
If none remain, it returns a checkable proof that no safe implementation exists.
Otherwise, the remaining ways describe exactly what the runtime may allow.
Our formal results cover Checkpoint and the six forms of Fork, Restore, and Merge, together with extensions, atomic enforcement, and the information every exact checker needs.
Lean mechanizes the finite checker and runtime invariant, and tests validate all six edit forms.
\ifarxiv
The source code, Lean proofs, and executable tests are available in the public GitHub repository at
\url{https://github.com/eunomia-bpf/agent-check-restore-safety}.
\fi
\end{abstract}

\begin{figure*}[t!]
  \centering
  \begin{tikzpicture}[
      font=\scriptsize,
      box/.style={draw,rounded corners=1pt,minimum height=5.5mm,
        align=center,inner sep=2pt},
      node/.style={draw,circle,minimum size=5.4mm,inner sep=0pt},
      arr/.style={-{Latex[length=1.5mm]},thick},
      stale/.style={-{Latex[length=1.5mm]},thick,dashed},
      lab/.style={font=\scriptsize\bfseries,anchor=east}
    ]
    \node[lab] at (0,2.15) {current workflow};
    \node[node] (b0) at (0.65,2.15) {$b_0$};
    \node[node] (al) at (1.85,2.15) {$x_a$};
    \node[box] (f) at (3.1,2.15) {choice\\Fork};
    \node[node] (ba) at (4.45,2.62) {$R$};
    \node[node] (bb) at (4.45,1.68) {$L$};
    \node[box] (ms) at (5.95,2.62) {select\\Merge};
    \node[node] (bl) at (7.25,2.62) {$b_R$};
    \node[box] (rl) at (8.35,1.68) {Restore $k$\\and continue};
    \node[node] (cl) at (9.75,1.68) {$x_a'$};
    \node[box] (jn) at (11.15,2.15) {join\\Merge};
    \draw[arr] (b0) -- node[above] {$k$} (al);
    \draw[arr] (al) -- (f);
    \draw[arr] (f) -- (ba);
    \draw[arr] (f) -- (bb);
    \draw[arr] (ba) -- (ms);
    \draw[arr] (ms) -- (bl);
    \draw[stale] (b0.north) -- (0.65,3.18) --
      node[above] {checkpoint copy: new call, same action}
      (8.35,3.18) -- (rl.north);
    \draw[arr] (rl) -- (cl);
    \draw[arr] (bl) -- (jn);
    \draw[arr] (cl) -- (jn);
    \draw[stale] (bb) to[bend left=17]
      node[below] {other branch disabled} (ms);

    \node[lab] at (0,0.65) {recorded calls and authorization order};
    \node[node] (r) at (0.65,0.65) {$r$};
    \node[node] (a) at (3.05,0.65) {$a$};
    \node[node] (pa) at (6.05,0.65) {$p_R$};
    \node[node] (s) at (10.2,0.65) {$s$};
    \draw[arr] (r) -- node[above] {append only} (a);
    \draw[arr] (a) -- (pa);
    \draw[arr] (pa) -- (s);
    \node[align=center] at (9.75,1.12)
      {\(\Alias(x_a',d_a)\)\\reuses earlier authorization};

    \node[lab] at (0,-0.75) {rules in force};
    \node[box] (e0) at (0.85,-0.75) {$\eta_0$};
    \node[box] (ef) at (3.7,-0.75) {$\eta_F:\ L,R$};
    \node[box] (ea) at (6.8,-0.75) {$\eta_R$};
    \node[box] (ej) at (11.15,-0.75) {$\eta_J$};
    \draw[arr] (e0) -- node[above] {atomic swap} (ef);
    \draw[arr] (ef) -- node[above] {atomic swap} (ea);
    \draw[arr] (ea) -- node[above] {atomic swap} (ej);
  \end{tikzpicture}
  \caption{Fork, Restore, and Merge change the current workflow but do not erase recorded calls.
  Restore creates a new call that refers to the same action as the original call.
  The atomic rule update orders authorization creation and reuse against the rule change.}
  \label{fig:two-histories}
\end{figure*}
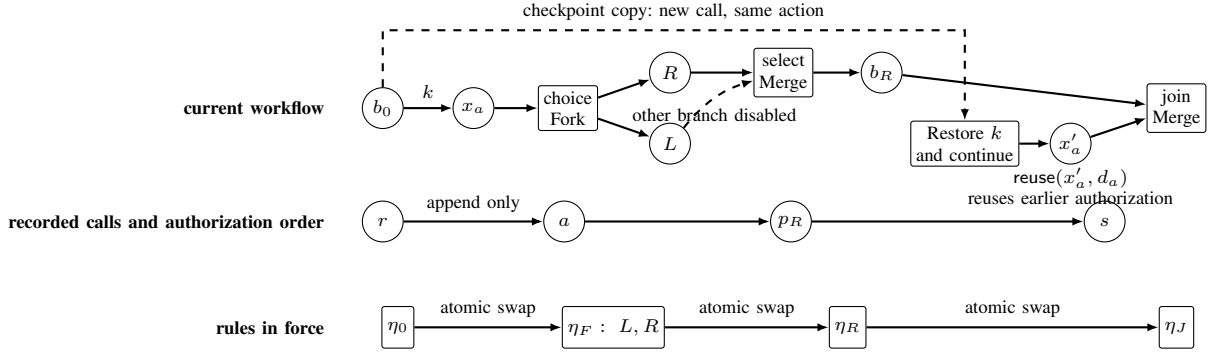

\section{Introduction}
\label{sec:introduction}

Tool-using Agent runtimes can change an execution while it is in progress.
They let an Agent Checkpoint an execution, Fork it into alternatives, Restore an earlier checkpoint, and Merge branches~\cite{codexappserver,claudesessions,claudecheckpoint,wang2026fork,wu2026crab,zhang2026agentlibos}.
These operations support exploration, recovery, and reuse without restarting the task.
We call all four operations \emph{execution edits}.
Checkpoint records the current execution for later use, while Fork, Restore, and Merge change what the Agent will do next.

Consider an Agent buying one laptop for a school.
Before asking the school to approve the purchase, it saves checkpoint \(k\).
The school approves one purchase, and the Agent later chooses supplier \(R\) and authorizes payment.
The Agent then restores \(k\) to compare supplier \(L\) and merges the \(L\) branch into the current execution.
If the runtime sees only the restored workspace, it may request the same approval again or continue with \(L\) even though payment to \(R\) is already authorized~\cite{zheng2026acrfence}.
Either behavior breaks the school's purchasing rules because the first repeats a one-time approval and the second combines incompatible supplier actions.

This example shows that the runtime cannot rely on the Agent's account of the execution.
The Agent may request an edit that omits an earlier action or a result the task still requires.
The trusted runtime instead determines what may happen next from an \emph{execution record} that the Agent cannot modify; in this example, the record remembers the one-time approval and the authorized payment to \(R\).
A \emph{tool action} is an external operation governed by policy, such as approval, payment, or shipment.
The runtime checks every tool action against policy.
The runtime makes its security decision when authorizing a tool request.
Sending the request later and recording its return cannot undo that decision.
The runtime stores that decision in an \emph{authorization record}, which later calls may reuse.
Our security goal is therefore to prevent an edit from authorizing a tool action forbidden by policy or discarding a result that the task still requires.
Meeting this goal requires four facts: (1) the task structure and its required results, (2) which calls refer to the same action, (3) earlier authorizations and call progress, and (4) which rules govern a call that overlaps the edit.
Each fact can change the correct answer.
The restored workspace, an execution trace alone, an authorization log alone, and the Agent's own description each omit information the checker may need.
Before an edit takes effect, the runtime records the finite set of relevant tool calls, their possible orders, and results.
The Agent's internal reasoning remains unchanged.

Existing approaches cover parts of this problem but do not derive its complete safety condition.
Agent systems provide branching, staged effects, transactions, or recovery checks~\cite{yu2026shepherd,li2025agentgit,langchain2026timetravel,patil2024goex,chang2025sagallm,mohammadi2026atomix,chen2026cordon,yang2026dart}.
Control and synthesis choose behavior after a model of possible executions and desired final states has been supplied~\cite{dequeiroz2005multitasking,malik2008generalised,finkbeiner2022live,aceto2018suppressions}.
Neither line of work derives from a recorded execution everything that a particular edit must preserve before it takes effect.

Our key idea is to derive from the execution record what each edit must preserve: earlier authorizations and every result that remains required.
Every result required before the edit must remain required, already be satisfied by the recorded execution, or be explicitly removed by the same authority that required it.
The Agent requests an edit under a registered edit rule.
The trusted runtime checks the recorded task structure and determines from that record, rather than from Agent text, what may happen next and which requirements carry over.
The Agent therefore cannot obtain a favorable answer by describing an execution that omits earlier actions or required results.

This idea creates three technical problems.
First, calls copied by Fork or Restore may refer to the same underlying action, so the runtime must distinguish a new authorization from reuse of an earlier authorization and return value.
Second, throughout every allowed execution, every result that is still required and still possible must retain a safe way to finish.
The checker must enforce this condition without forbidding additional safe behavior.
Third, the rules for an accepted edit must take effect atomically with calls that overlap the edit and remain valid if the runtime or machine stops unexpectedly and later recovers.

We give an exact checker that computes every safe way for a requested edit to proceed or proves that every possible way is unsafe.
For a requested edit, the runtime constructs every allowed complete execution from the current record.
It removes an execution if allowing its calls one at a time could leave a still-required result without a safe completion.
If none remains, the runtime returns an independently checkable proof that no runtime rule set can enforce the edit safely.
Otherwise, the remaining executions contain every safe way for the edit to proceed.
Before the edit takes effect, the runtime checks the current execution record again and atomically switches to the newly computed rules.
Each overlapping call is checked entirely under either the old or new rules.
\Cref{fig:two-histories} illustrates this process for an approval task containing Fork, Merge, Restore, and calls that overlap the atomic rule update.
After Restore, the repeated approval call reuses the earlier authorization record instead of authorizing the approval again.

\paragraph{Contributions}
\begin{enumerate}[leftmargin=*,itemsep=2pt,topsep=2pt]
\item \textbf{Safety model.}
We define Checkpoint, Fork, Restore, and Merge safety from the execution record rather than letting the Agent choose the edited execution.
The model preserves past actions and every still-required result and proves that exact checkers must distinguish answer-changing cases from all four record parts.

\item \textbf{Exact checking and atomic enforcement.}
We construct an executable checker that returns every safe way to proceed or a proof that none exists.
Checkpoint has an exact check, and the checker accepts each of the six forms of Fork, Restore, and Merge exactly when the safety condition is enforceable.
Our formal runtime protocol rechecks the record, makes accepted rules take effect atomically, preserves safety across repeated edits and restarts, and has the same Agent-visible behavior as an ideal atomic machine.

\item \textbf{Mechanization and executable validation.}
Six Lean modules mechanize registration, all six edit forms, the atomic rule update, concurrent calls, and preservation across finite executions.
Nineteen paper-specific tests and measurements over 2--128 results exercise the executable checker.
\end{enumerate}

\section{Execution-Edit Safety}
\label{sec:lifecycle}

Imagine an Agent buying one laptop for a school.
It reserves the budget, saves checkpoint \(k\), receives one purchase approval, and then compares suppliers \(L\) and \(R\).
A careless Restore or Merge could request approval twice, pay both suppliers, or permit shipment before the selected supplier is paid.
A safe execution approves the purchase once, pays exactly one supplier, and permits shipment only after that payment.

The resulting \emph{workflow} records which calls may occur, in what order, and which results count as completion.
In this example, it authorizes at most one purchase.
The Agent reserves the budget and saves checkpoint \(k\) before approval.
The approval is then authorized once, creating an authorization record that Restore does not delete.
The Agent next requests a Fork between suppliers \(L\) and \(R\).
The workflow allows two results: order from \(L\) or order from \(R\).
In either result, payment authorization must precede shipment authorization:
\begin{equation}
\label{eq:running-results}
\begin{split}
 o_L&:\ x_L^p\prec x_L^s,\\
 o_R&:\ x_R^p\prec x_R^s ,
\end{split}
\end{equation}
where \(x_i^p\) and \(x_i^s\) are payment and shipment calls.
A \emph{workflow call} \(x\) is one place in this structure where the runtime can invoke a tool.
The function \(q_{\rm act}(x)=d\) records which tool action the call refers to.
Calls that refer to the same action share one authorization record.

In this example, both shipment calls authorize the same one-order shipment:
\[
 q_{\rm act}(x_L^s)=q_{\rm act}(x_R^s)=d_s.
\]
The two payment calls instead refer to different actions \(d_L^p\) and \(d_R^p\).

The authorization log \(\Delta\) records newly authorized actions in order.
The notation \(\mathsf{lab}(\Delta)\) keeps only their authorization labels.
Budget reservation contributes \(r\), approval contributes \(a\), the two payments contribute \(p_L\) and \(p_R\), and shipment contributes \(s\).
At the Fork request, \(\mathsf{lab}(\Delta)=ra\).
The policy allows every partial authorization sequence of \(rap_Ls\) and \(rap_Rs\), but never both payments or a payment before approval.
Although the supplier results are exclusive, both must still have a safe way to finish when the Fork is checked.

\subsection{Accepted and Rejected Forks}
The Agent sends \(u=\mathsf{ForkChoice}(b,\sigma)\), which names a current branch and a registered edit rule.
Here, the rule identified by \(\sigma\) turns branch \(b\) into the exclusive supplier-\(L\) and supplier-\(R\) copies shown below.
The trusted runtime constructs the edited workflow from the execution record.
It also derives the results that remain required and the action referred to by every workflow call.

Write \(M_o\) for the ways to complete result \(o\) while obeying the policy.
Each result has one allowed call order:
\[
\begin{aligned}
 z_L&=\Fresh(x_L^p,d_L^p)\Fresh(x_L^s,d_s),\\
 z_R&=\Fresh(x_R^p,d_R^p)\Fresh(x_R^s,d_s).
\end{aligned}
\]
Thus \(M_{o_L}=\{z_L\}\) and \(M_{o_R}=\{z_R\}\).
These sequences contain the calls after the edit, while the existing log contributes the earlier authorization sequence \(ra\).
The verdict is \(\mathsf{Accepted}\): both supplier results remain possible, the calls performed so far always obey policy, and each result can be completed.
If policy excludes \(rap_Rs\), then \(M_{o_R}=\varnothing\), and the Fork is rejected because pruning supplier \(R\) would change its meaning.
Rejection returns a proof listing the order in which completions were removed and the reason for each removal, rather than letting the Agent weaken the requested edit.

\section{Implications of Execution-Edit Safety}
\label{sec:implications}

This section explains six consequences of our formal results for runtime design.
Here, a \emph{result} means one way the task can finish, such as completing the laptop purchase through supplier \(L\) or \(R\).
It is not a value returned by a tool.
A result remains required until the execution record shows that it has finished or that the party requiring it allowed its removal.

\paragraph{A smaller workflow can be less safe}
In the laptop example, the current workflow allows the school to buy through supplier \(L\) or supplier \(R\).
Suppose the Agent edits it to keep only \(L\).
Every remaining call may obey the purchase policy, but the edit removes the still-required \(R\) result.
The edit is safe only if the record shows that the \(R\) purchase has finished or that the school allowed its removal.
The runtime must therefore derive what the edit must preserve from the execution record, not from the proposed workflow.

\paragraph{Two safe plans may have no safe way to run together}
Suppose a release workflow must support both normal release and emergency release.
The normal plan obtains approval \(y\) before deployment \(x\), producing \(yx\).
The emergency plan deploys first and then files an emergency report \(z\), producing \(xz\).
The policy \(\mathcal A=\{\epsilon,x,y,xz,yx\}\) permits both orders, so either plan is safe when checked alone.
In the combined workflow \(X\otimes(Y\oplus Z)\), both plans contain \(x\).
Allowing \(x\) first leaves the normal plan able to finish only in the forbidden order \(xy\), while blocking \(x\) makes the emergency plan impossible.
The runtime cannot make either choice while keeping both plans safe, so the exact checker returns \(\mathsf{Reject}\).

\paragraph{An edit can change from Accepted to Reject without a new authorization}
An edit verdict applies only to the execution record that the checker examined.
In \cref{fig:two-histories}, Restore copies an approval call as \(x_a'\).
The original call and \(x_a'\) refer to the same authorized action \(d_a\).
Processing \(x_a'\) therefore adds no authorization-log entry, but it still moves the restored branch past approval.
A concurrent Merge or Restore that was Accepted before this call may be rejected afterward.
If the new verdict is Reject, the edit must not take effect in that execution state.
Requesting authorization for \(d_a\) again cannot undo this progress because the runtime reuses \(d_a\)'s existing authorization record.
The runtime must discard the earlier verdict and check the edit against the new record.
A later check can change its answer after a relevant fact changes, such as completing a required result or authorizing its removal, but not by resetting recorded history.

\paragraph{Saving every call and authorization is still not enough}
Suppose the runtime retains all workflow calls, actions, and authorization records but loses the links among them.
It can no longer tell whether a restored ``approve purchase 42'' call reuses purchase 42's authorization or needs a new one.
The tool name and call text cannot recover this information.
\Cref{thm:necessary-state} proves that two records containing the same calls, actions, and authorizations can require opposite verdicts when these links differ.
An exact runtime must therefore record which action each workflow call and authorization refers to.

\paragraph{There is one largest safe set of call orders}
Safety could in principle admit several incomparable rule sets, each allowing call orders that the others block.
\Cref{sec:model} instead proves that the checker returns one safe execution set containing every other safe execution set.
The generated runtime rules therefore combine all call orders that can coexist without violating policy or preventing a required result from finishing.

\paragraph{Exactness survives later calls, edits, and restarts}
The exact guarantee is not limited to the first accepted edit.
After allowed calls have run, checking the remaining workflow again gives the same safe ways to finish as the installed rules.
Later Fork, Restore, and Merge requests, rule updates, stops, and restarts preserve the same guarantee.
Every later edit is therefore checked under the same exact safety condition as the first.

\Cref{sec:model} formalizes the execution record and safety goal.
\Cref{sec:exact-check} defines the checker, and \cref{sec:enforcement,sec:results} show how the runtime installs its rules atomically and preserves the guarantees above.

\section{Model and Security Objective}
\label{sec:model}

The model leaves the Agent's internal reasoning unrestricted and requires every tool call to pass through the trusted runtime.
Each call identifies its registered place in the workflow, presents an authorization token, and submits a concrete tool request.
The runtime either creates a new authorization, reuses an authorization already recorded in the log, or rejects the call.
We analyze calls in independently enforceable groups.
Calls governed by one policy, one authorization log, and one atomic rule change form a \emph{policy domain}.
For example, the laptop approval, payment, and shipment calls form one policy domain when they share the school's purchase policy and authorization log.
Policy domains with disjoint actions, logs, and policies compose independently.
Every individual state is finite, while the theory covers executions of arbitrary finite length.
The \emph{execution record} stores the workflow and checkpoints, which calls refer to the same action, earlier authorizations and call progress, queued requests, and the rules currently in force.
The trusted runtime protects this record from Agent modification.
The model has one security objective.
Every admitted execution must obey policy, remain completable, and preserve a safe completion for every still-required result.
The remainder of this section defines the execution facts needed to state that objective precisely.

\subsection{Workflow results and call order}

\paragraph{Basic objects}
We capitalize \emph{Agent} for the untrusted principal requesting execution edits.
A \emph{workflow call} \(x\) is one place where the workflow can invoke a tool.
The function \(q_{\rm act}(x)=d\) records the tool action referred to by that call.
Several workflow calls may refer to the same action \(d\), in which case they share one authorization record.
The first authorized call creates that record, and calls copied by Fork or Restore later reuse it.
In this paper, a workflow call is processed when the runtime either authorizes its tool action for the first time or reuses that action's existing authorization record.
Processing does not mean that the remote action has completed: sending a newly authorized request and recording its returned value happen later, and neither adds another processed workflow call.
The execution record also stores where every copied workflow region came from.
\(T,\Delta,\chi\) denote the current workflow, authorization log, and ordered record of processed workflow calls.

A pomset \(p=(X_p,\prec_p)\) is a finite set of workflow calls with a strict partial order.
Its linearizations \(\Lin(p)\) are the words containing each member of \(X_p\) once and respecting \(\prec_p\).
A \emph{workflow contract} \(\mathcal C=(O,\{p_o\}_{o\in O})\) assigns one pomset to each member of a finite nonempty result set \(O\).
For the laptop, the contract has results \(o_L,o_R\), and \(p_{o_L}\) orders the \(L\)-payment call before the \(L\)-shipment call.
Each index \(o\) names one allowed result.
A result is complete once every workflow call in its pomset has been processed in an allowed order.
Two results remain distinct even if the calls still required for them later become identical.
Write \(X_{\mathcal C}=\bigcup_{o\in O}X_{p_o}\) and \(\mathsf{CL}(\mathcal C)=\bigcup_{o\in O}\Lin(p_o)\).
We require the runtime to distinguish completed work from work that has more calls.
For sequences \(v,w\), write \(v\preceq w\) when \(w=vs\) for some sequence \(s\).
\[
 v,w\in\mathsf{CL}(\mathcal C)\ \land\
 v\preceq w \quad\Longrightarrow\quad v=w .
\]
This condition permits incomparable traces and different results with the same trace, so it does not require deterministic results.
It requires a registered workflow call whenever more work can follow the same sequence of processed calls.
The interface represents a choice between stopping and continuing by a registered terminal workflow call checked by the trusted runtime.

Contracts use the three partial constructors below.
Each constructor requires operands with disjoint workflow calls and preserves the condition that no complete call sequence can be properly extended into another.
\begin{equation}
\label{eq:contract-algebra}
\begin{aligned}
 O_{\mathcal C\oplus\mathcal D}
  &= (\{L\}\mathbin{\times}O_{\mathcal C})
     \uplus(\{R\}\mathbin{\times}O_{\mathcal D}),\\
 p_{(L,o)}&=p_o,\qquad p_{(R,v)}=p_v,\\
 O_{\mathcal C\otimes\mathcal D}
  &=O_{\mathcal C}\mathbin{\times}O_{\mathcal D},\\
 p_{(o,v)}&=p_o\uplus p_v,\\
 O_{\mathcal C\triangleright\mathcal D}
  &=O_{\mathcal C}\mathbin{\times}O_{\mathcal D},\\
 p_{(o,v)}&=p_o\triangleright p_v.
\end{aligned}
\end{equation}
Before adding a registered or copied operand, the trusted runtime deterministically assigns fresh workflow-call IDs and extends \(\rho\) while retaining which tool action each call refers to.
Well-formedness already makes carried operands disjoint.
In \cref{eq:contract-algebra}, \(\uplus\) adds no cross-order and \(\triangleright\) orders every event of \(p\) before every event of \(q\), while choice uses a tagged union of result indices and parallel and sequence use Cartesian products.
Consequently, choice retains each arm's full indexed family, parallel pairs results while permitting every causal linearization, and sequence adds a barrier.
After call sequence \(w\), \(\mathcal C/w\) retains each result whose pomset has a linearization beginning with \(w\), then removes the workflow calls in \(w\) and their order edges.
The expression is undefined if no result remains.
This operational ``remaining contract'' preserves result and workflow call identity.
For example, after the \(L\)-payment call, the remaining contract for \(o_L\) contains only the \(L\)-shipment call.

\begin{lemma}[Completion remains unambiguous]
\label{lem:observable-completion}
For every defined \(\mathcal C/w\), no complete sequence can be properly extended into another, and either every retained pomset is empty or none is empty.
\end{lemma}
\begin{proof}
If a remaining sequence \(v\) can be properly extended into \(v'\), then \(wv\) can be properly extended into \(wv'\) in \(\mathsf{CL}(\mathcal C)\).
If one retained pomset is empty and another is nonempty, choose a nonempty linearization \(v\) of the latter.
Then \(w\) can be properly extended into \(wv\).
Both contradict the condition that no complete sequence can be properly extended into another.
\end{proof}

\subsection{Execution Record}

\subsubsection{Recorded structure and current workflow}
The execution record separates immutable past events, the current workflow, registered edit rules, call identities, progress, and the active rule version.
The workflow part of the execution record is
\begin{equation}
\label{eq:history-state}
 H=(\omega,G,T,K,\Sigma_\zeta,\rho,\beta,\chi,\mathsf{ver},\eta).
\end{equation}
\(\omega\) identifies the policy domain, and \(\eta\) identifies the version of the rules currently in force.
\(G\) records the finite sequence of workflow versions.
Its node IDs, registered contracts, signed result requirements, and Fork, Restore, and Merge records do not change.
Each branch node \(b\) carries an immutable registered contract \(\Gamma_b\).
Its ordered progress record \(\chi_b\) lists the workflow calls already processed on that branch.
The immutable records form a directed acyclic graph (DAG) whose edges record typed Fork, Restore, and Merge operations.
\(T\) is the current workflow.
\[
 T ::= b{:}\mathcal C
     \mid T\oplus_g^s T
     \mid T\parallel_g T
     \mid \mathsf{join}(T,T)
     \mid T;T .
\]
\(s\in\{\mathsf{open},L,R\}\) records an unresolved choice or the arm that made recorded progress, and \(g\) groups exclusive or jointly active siblings.
\(\sem{T}\) uses both arms of an open choice and only the selected arm thereafter.
It interprets parallel and \(\mathsf{join}\) with \(\otimes\), and sequence with \(\triangleright\).
The \(\mathsf{join}\) marker records a completed Merge and prevents the same group from being merged again.

\(T\) is \(\mathsf{complete}\) when every pomset in \(\sem{T}\) is empty.
\Cref{lem:observable-completion} ensures that the remaining workflow cannot contain both empty and nonempty pomsets.
Set \(\mathsf{heads}(b{:}\mathcal C)=\{b\}\).
An open choice makes both arms current, a selected choice makes its winner current, parallel and \(\mathsf{join}\) make both arms current, and sequence makes its left operand current until completion and then its right.
Processing a workflow call updates the work remaining at its leaf without deleting the recorded structure.
Because sequence checks \(\mathsf{complete}\), it retains completed results from its left operand, so an isolated completed leaf remains available for a later Fork or Restore.
Joining removes \(g\) immediately but keeps its remaining work behind a barrier until both arms complete.
A workflow call is enabled when it is minimal in some pomset of \(\sem{T}\).
Disjoint workflow calls and unique context decomposition give every enabled \(x\) a unique current leaf \(\mathsf{leaf}_T(x)=b\).
For \(a\in\{\Fresh(x,d),\Alias(x,d)\}\), define the paired execution update
\[
 \mathsf{HStep}((G,T),a)=
 \bigl(G\mathbin{\|}(b,a),\,T/(x,a)\bigr),
 \qquad b=\mathsf{leaf}_T(x),
\]
where \(G\mathbin{\|}(b,a)\) appends progress record \((b,a)\), and \(T/(x,a)\) removes the processed call from the current leaf and records any selected choice.
Write \(\mathsf{HRes}((G,T),r)\) for iterating this update over a resolved word \(r\).

\(\mathsf{addr}(T)\) is the set of branch and group IDs that an edit may currently name.
It includes each current group and the current part of a sequence.

\(K\) maps checkpoint IDs to immutable earlier branch records.
Write \(\mathcal C_k=\mathsf{contract}(K(k))\) for the remaining contract stored by checkpoint \(k\).
\(\Sigma_\zeta\) is an immutable, versioned registry of typed edit rules.
\(\rho\) records, for each workflow call, its action, authorization token, source region, scope, canonical tool request, and request digest.
The source region identifies the registered workflow region from which the call originates, while the scope limits where its authorization token may be used.
For example, a copied payment call retains the source region of the registered payment step and its purchase-domain scope.
Write \(q^{\rho}_{\rm act}(x)\) for the action field of \(\rho(x)\).
For \(d\in\operatorname{im}(q^{\rho}_{\rm act})\), let \(\mathsf{inv}_\rho(d)\) be the common canonical-request field of calls \(x\) with \(q^{\rho}_{\rm act}(x)=d\).
Let \(\mathcal U\) be the authorization-label alphabet.
Independently, \(\ell(d)\in\mathcal U\) is the authorization label certified for tool action \(d\) by the registered call model.
For each current workflow call \(x\), \(\beta(x)=(j,\eta)\) records the rule entry \(j\) and version \(\eta\) used to check it.
\(\chi\) records processed workflow calls in order, including whether each created or reused an authorization record, and \(\mathsf{ver}\) is the state version.
\(\chi\) records authorization and workflow progress, while the remote service records completion and the returned value.

\(\mathsf{Checkpoint}(k,b)\) creates a checkpoint for a current branch \(b\) under a fresh checkpoint ID \(k\).
A trusted transaction stores a signed copy of branch \(b\)'s remaining contract \(\mathcal C_b\), ordered progress record \(\chi_b\), relevant calls, and result requirements.
The checkpoint check requires \(\mathcal C_b=\Gamma_b/\mathsf{raw}(\chi_b)\), and the transaction advances \(\mathsf{ver}\) without changing the current workflow, authorization records, \(\beta\), the active rule version, or the authorization sequence.
This internal transition changes \(H\), so a rule update computed from the earlier record will fail its final check.
Restore accepts only records created by this rule.

\subsubsection{Well-formedness}

\begin{definition}[Well-formed execution record]
\label{def:well-formed-history}
\(H\) is well formed when the following groups hold.
\begin{enumerate}[leftmargin=*,itemsep=0pt,topsep=1pt]
\item \emph{Structure.}
\(G\) is acyclic, and branch and group IDs are globally unique across the recorded DAG and current workflow.
Every ID that an edit may name has a unique context decomposition, and every node in \(\mathsf{heads}(T)\) is currently maximal.
Every checkpoint names an immutable node and stores copies of its remaining contract and ordered progress record.
Every original or edit-introduced result names exactly one policy authority in a signed record.
The trusted runtime preserves that record when it carries or copies the result or stores it in a checkpoint.
\item \emph{Identity and progress.}
For every recorded workflow call \(x\), \(q^{\rho}_{\rm act}(x)=q_{\rm act}(x)\).
Workflow-call IDs and the \(j\) identifiers in \(\beta\) are unique, while workflow calls that refer to the same action agree on which trusted runtime authorized them, their scope, canonical request, digest, and source region.
Processed workflow calls are downward closed in their pomset, consistent with every choice status, and absent from the remaining workflow.
\item \emph{Current workflow.}
\(\operatorname{dom}(\beta)=X_{\sem T}\), and \(\beta\) maps every current workflow call \(x\) to exactly one entry \(j_x\) under the domain's current rule version \(\eta\).
Every current leaf \(b{:}\mathcal C\) equals its immutable registered contract \(\Gamma_b\) after removing the workflow calls recorded in \(\chi_b\); the caller cannot replace it with a smaller result set.
Every constructor in \(T\) is defined, so no complete sequence in any subtree can be properly extended into another.
\item \emph{Edit rules and domain.}
Every edge and introduced contract follows one immutable edit rule at version \(\zeta\).
Every current, checkpoint, or rule-introduced workflow call, action, source region, and scope belongs to \(\omega\).
Tool actions, authorization logs, and policies of distinct domains are disjoint, and an edit never crosses a domain.
\end{enumerate}
The policy domain has exactly one active rule version, namely the version recorded in \(H\).
\end{definition}

\subsubsection{Authorization and completion}

The remaining parts of the execution record are policy version \(\kappa\), a regular authorization policy \(\mathcal A\subseteq\mathcal U^*\) closed under truncation, authorization log \(\Delta\), and queue \(\mathsf{out}\) of requests waiting to be sent.
Authorization records can only be appended.
A completion update advances an authorization record's status and stores the returned value.
Each authorization record binds its action, immutable canonical tool request, creator workflow call, and authorization label.
Each entry in \(\chi\) points either to the authorization record it created or to the earlier record it reused, so \((\chi,\Delta)\) determines retry and return lookup.
\(\mathsf{lab}(\Delta)\) is the sequence of authorization labels on calls that create authorization records.
Restore may reconstruct workspace data but leaves \(\Delta\) unchanged.
A returned value may lead the Agent to request another registered edit, which the runtime checks from the updated execution record.

The execution record \((\kappa,H,\Delta,\mathsf{out},\mathcal A)\) is \emph{well formed} when \(H\) is well formed, \(\mathsf{lab}(\Delta)\in\mathcal A\), each action has at most one authorization record, each processed call names the record it created or reused, the authorization log, request queue, and policy belong to \(\omega\), and the queue preserves the order of new authorizations.
All definitions and results below quantify over such well-formed execution records.

\subsection{Adversary and trusted boundary}

The adversary controls Agent output, edit requests, tool arguments, copied workspace, retries, stale authorization tokens, and scheduling, including races between calls and edits.
It cannot forge registries, the checkpoint directed acyclic graph (DAG), the record of which calls refer to the same action, the policy, or the authorization log.
The policy authority signs policies, result requirements, and authorized removals.
In the laptop example, the school is the policy authority: it signs the one-purchase requirement and any permission to remove a supplier result.
The trusted runtime registers tool calls, actions, and authorization tokens, and it derives and checks each edit.
The trusted runtime checks the answer again before atomically changing the rules in force.
The policy authority, trusted runtime, registries, authorization log, proof verifier, and transaction mechanism form the trusted computing base (TCB).
Within one policy domain, the TCB checks every tool call and commits every state change and edit answer returned to the Agent through a linearizable transaction.
After an unexpected stop and restart, each transaction appears entirely before or entirely after its linearization point.
Before the checker reads the execution record, the trusted runtime may instantiate finitely many registered templates for tool calls.
For tool request \(m\), it computes canonical request \(m^\circ=\mathsf{canon}_\kappa(m)\), reuses an action certified for \(m^\circ\) or allocates a fresh one, and signs \((x,j,d,\mathsf{scope},\mathsf{digest}(m^\circ),\mathsf{origin})\) into the authenticated registry.
The resulting finite model \(\mathcal R\) contains the workflow's tool calls and remains unchanged while the checker evaluates that execution record.
The runtime rejects any call that does not match the registered model.
One execution record outside the Agent's control contains workflow \(\mathcal W\), finite model \(\mathcal R\), edit rules, request-normalization rules, policy, checkpoints, authorization log, and queued requests.
From that record and an Agent request naming an edit and its objects, the trusted runtime derives the edited workflow, required results, authorization decision for each call, safe execution set or rejection proof, and the rule entry and authorization token for each remaining call.
Every derived object therefore comes from the execution record rather than being chosen by the Agent.

Registration turns the tool-call behavior of a typed Agent workflow into the finite model \(\mathcal R\).
\(\mathsf{CheckReg}\) verifies exact correspondence of results, call order, actions, and first-time authorization records.
The model stores these traces as runs labeled by result, written \(\mathsf{BRuns}(\mathcal R)\), together with a map \(\lambda\) from source traces.
Write \(\mathcal W\sqsubseteq_\lambda\mathcal R\) when \(\lambda\) is a bijection preserving results, which calls refer to the same action, call order, result authority, and authorization-log updates.
\Cref{def:registration-refinement,lem:registration-reflection} give the exact finite check and proof.
Registered contracts enter through \(\mathsf{CheckReg}\), and an idempotent or queryable remote service supplies signed completion records for its actions.

\subsection{Security Objective}

The security objective is fixed independently of the checker that will implement it.
For an edit derived from the execution record, a safe execution set satisfies four requirements.
First, it preserves every source result that is neither already satisfied nor explicitly authorized for removal, and every allowed execution belongs to the derived target workflow.
Second, every allowed partial execution obeys the policy after accounting for authorization records already present in the append-only log.
Third, each allowed partial execution extends to a complete allowed execution.
Fourth, every result still compatible with an allowed partial execution retains such a completion.
For example, at the laptop Fork, the set must retain both supplier results unless the school authorizes removing one, admit only approval--payment--shipment orders allowed by policy, and preserve a completion from every admitted partial execution for each supplier result still compatible with that execution.
Thus the checker cannot call the edit safe merely because the surviving \(L\) execution obeys policy after all completions of \(R\) have been pruned.
The last requirement rejects the two-result example in \cref{sec:implications}, where separately safe results cannot remain safe under one set of runtime rules.
\Cref{def:safe-execution-set} states these requirements over executions labeled by result.
\Cref{lem:largest-safe-family} then proves that the checker computes the largest safe execution set.

\section{Exact Checking of Execution Edits}
\label{sec:exact-check}

The checker first derives the edited workflow prescribed by the execution record and the requested edit.
It then decides whether each workflow call creates or reuses an authorization record and removes precisely those complete executions that violate the security objective after some partial execution.

\subsection{Constructing the Edited Workflow}

\subsubsection{Checking an edit rule}

Fork, Restore, and Merge have the six disjoint constructors of \(\mathsf{Edit}_6\) shown in \cref{tab:history-rules}.
\(\mathsf{Checkpoint}\), the initial rule update, later rule registration, and tool calls use their own transitions.

An Agent request \(u\) names an edit kind, current object IDs, and edit-rule ID \(\sigma\).
Registry \(\Sigma_\zeta\) records the expected source and checkpoint patterns for \(\sigma\), any workflow steps the edit adds, the source region for each copy, and the authority that signs each introduced result.

The checker constructs explicit preservation maps for every region that the rule carries or copies.
Consider a source contract \(\mathcal C\), target region \(\mathcal D\), and parent map \(\mu:X_{\mathcal D}\to X_{\mathcal C}\).
Write \(\mathcal C\preccurlyeq_\mu\mathcal D\) when \(\mu\) is a global bijection and preserves actions and source regions.
The relation also requires a surjection \(\pi:O_{\mathcal D}\twoheadrightarrow O_{\mathcal C}\) that satisfies the conditions below for every target result \(v\).
\[
\begin{aligned}
 x\in X_{p_v}
   &\Longleftrightarrow \mu(x)\in X_{p_{\pi(v)}},\\
 x\prec_v y
   &\Longleftrightarrow \mu(x)\prec_{\pi(v)}\mu(y).
\end{aligned}
\]
Because \(\mu\) is one global bijection and preserves which results contain each call, it cannot split a source call shared by several results into separate target calls.

For candidate \(T'\), let \(R_u\) index the disjoint regions carried or copied inside the replaced subtree.
\(\mathcal P_u=\{(\mathcal C_r,\mathcal D_r,\mu_r,\pi_r)\}_{r\in R_u}\) collects these preservation maps.
A separate identity map verifies that the surrounding workflow preserves syntax, results, workflow calls, actions, source regions, and order.
Writing \(X_{\rm ctx}\) for its workflow calls and \(X_{\rm added(u)}\) for workflow calls added by the rule, the checker verifies the following disjoint cover.
\[
 X_{\sem{T'}}=X_{\rm ctx}\uplus
   \biguplus_{r\in R_u}X_{\mathcal D_r}\uplus X_{\rm added(u)} .
\]
The edit rule and surrounding workflow induce \(\mathsf{pr}^u_r:O_{\sem{T'}}\rightharpoonup O_{\mathcal D_r}\).
For choice, this map is defined only in the arm containing \(r\), while product and sequence select the corresponding indexed component.
Define
\[
\begin{aligned}
 \mathsf{EditedRequired}_u&=\{(r,o)\mid r\in R_u,\ o\in O_{\mathcal C_r}\},\\
 \mathsf{Covers}_u(t,r,o)&\Longleftrightarrow
   \exists v\in O_{\mathcal D_r}.\quad
   \mathsf{pr}^u_r(t)=v\ \land\ \pi_r(v)=o .
\end{aligned}
\]
Let \(O_{\rm unchanged}\) be the results outside the edited region, and let \(\iota^u_{\rm ctx}:O_{\rm unchanged}\hookrightarrow O_{\sem{T'}}\) be their unchanged embedding into the target.
Define the results that remain required and the target results that cover them by
\begin{equation}
\label{eq:target-required}
\begin{aligned}
 \mathsf{StillRequired}_u&=\mathsf{EditedRequired}_u\\[-3pt]
   &\quad\uplus(\{\mathsf{ctx}\}\times O_{\rm unchanged}),\\
 \mathsf{Covers}_u(t,(\mathsf{ctx},o))
   &\Longleftrightarrow\\[-3pt]
   &\quad t=\iota^u_{\rm ctx}(o).
\end{aligned}
\end{equation}
For \(a=(r,o)\in\mathsf{EditedRequired}_u\), \(\mathsf{Covers}_u(t,a)\) denotes the relation above.
The checker also requires \(\forall a\in\mathsf{StillRequired}_u.\ \exists t\in O_{\sem{T'}}.\ \mathsf{Covers}_u(t,a)\): every required source result has a complete target result, shared workflow calls remain shared, and inherited calls remain distinct from steps added by the edit rule.

The edit-rule check uses the following judgment.
\[
\Sigma_\zeta;H\vdash_{\rm edit}u\Rightarrow(T',\mathcal P_u)
\]
It holds exactly when every named object belongs to \(\mathsf{addr}(T)\), every introduced object belongs to \(\omega\), the source pattern matches, the preservation maps and disjoint cover hold, every still-required result is covered, every removal carries the proper authority, and \(\sem{T'}\) is defined.
The policy authority may authorize removal of the unselected arm of a typed \(\mathsf{MergeSelect}\) or the current branch of \(\mathsf{RestoreReplace}\), using the signed sets defined below.
\subsubsection{Preserving Required Results}

The registered edit rule determines every result that the edit must account for.
Fork accounts for both copies, Restore for the current and checkpoint branches, \(\mathsf{MergeSelect}\) for the selected and unselected arms, and \(\mathsf{MergeJoin}\) for both arms.
Results outside the edited region remain unchanged.
\Cref{def:source-results} defines the results required before the edit, those already satisfied, and those whose removal the policy authority signed.
The checker accepts the edit rule when the target-derived set from \cref{eq:target-required} satisfies
\[
\begin{aligned}
\mathsf{StillRequired}_u
 &=\mathsf{RequiredBefore}(H,u)\\[-2pt]
 &\quad\setminus\bigl(
   \mathsf{Satisfied}_{H,u}\uplus
   \mathsf{AuthorizedRemoval}_{H,u}\bigr),
\end{aligned}
\]
and hence the exact disjoint partition
\[
\begin{aligned}
\mathsf{RequiredBefore}(H,u)
 &=\mathsf{StillRequired}_u\\[-2pt]
 &\quad\uplus\mathsf{Satisfied}_{H,u}\\[-2pt]
 &\quad\uplus\mathsf{AuthorizedRemoval}_{H,u}.
\end{aligned}
\]
This partition accounts for every source result and prevents an Agent request from dropping one without authorization.
The preservation maps retain actions, source regions, and order, while the rule update retains the records for satisfied results and authorized removals.

Let \(\mathcal E[-]\) be the current workflow with one editable region replaced by a hole.
\(\mathsf{carry}\) preserves branch, result, and workflow-call IDs, pomset order, actions, signed result authorities, base contracts, ordered progress records, and nested choices.
\(\mathsf{clone}_r\) assigns fresh branch, result, and workflow-call IDs while preserving pomset order, actions, scope, digest, source regions, and the authority for each result.
Each copied branch starts from its registered remaining contract with an empty progress record.
Branches added by an edit rule likewise start from their registered contract and an empty progress record.
The contracts \(E_\sigma^L,E_\sigma^R,E_\sigma^J\) are additional workflow steps derived from the rule registry.
Define
\[
 L_\sigma=\mathsf{clone}_L(\mathcal C)\triangleright E_\sigma^L,\qquad
 R_\sigma=\mathsf{clone}_R(\mathcal C)\triangleright E_\sigma^R .
\]
New identifiers and the target rule version \(\eta^+\) are deterministically allocated from \((\omega,\zeta,\mathsf{ver},u,\mathsf{role})\).
\(\mathsf{role}\) identifies which operation argument receives each new ID.
Write \(\mathsf{Copy}_r(\mathcal C)\) for \((\mathcal C,\mathsf{clone}_r(\mathcal C),\mu_r,\pi_r)\), and \(\mathsf{Keep}(\mathcal C)\) for \((\mathcal C,\mathsf{carry}(\mathcal C),\mu,\pi)\).
Each abbreviation asserts the corresponding \(\preccurlyeq_\mu\) relation.
For a current workflow \(T_0\), \(\mathsf{Keep}(T_0)\) abbreviates \(\mathsf{Keep}(\sem{T_0})\) and has target contract \(\sem{\mathsf{carry}(T_0)}\).
Every \(b\) or \(g\) named below must lie in \(\mathsf{addr}(T)\).

\begin{table*}[t]
\centering
\caption{The six forms of Fork, Restore, and Merge.
\(E_\sigma^\bullet\) is derived from the registered edit rule in \(\Sigma_\zeta\).}
\label{tab:history-rules}
\scriptsize
\setlength{\tabcolsep}{4pt}
\begin{tabular}{@{}p{0.18\textwidth}p{0.22\textwidth}
                    p{0.37\textwidth}p{0.15\textwidth}@{}}
\toprule
Operation \(u\) & Required current shape & Derived target shape &
Checked preservation\\
\midrule
\(\mathsf{ForkChoice}(b,\sigma)\) &
\(\mathcal E[b{:}\mathcal C]\) &
\(\mathcal E[b_L{:}L_\sigma\oplus_g^{\mathsf{open}}
 b_R{:}R_\sigma]\) &
\(\{\mathsf{Copy}_L(\mathcal C),\mathsf{Copy}_R(\mathcal C)\}\)\\
\(\mathsf{ForkParallel}(b,\sigma)\) &
\(\mathcal E[b{:}\mathcal C]\) &
\(\mathcal E[b_L{:}L_\sigma\parallel_g b_R{:}R_\sigma]\) &
\(\{\mathsf{Copy}_L(\mathcal C),\mathsf{Copy}_R(\mathcal C)\}\)\\
\(\mathsf{RestoreReplace}(b,k,\sigma)\) &
\(\mathcal E[b{:}\mathcal C],\ k\in\operatorname{dom}K\) &
\(\mathcal E[b_k{:}\mathsf{clone}_k(\mathcal C_k)]\) &
\(\{\mathsf{Copy}_k(\mathcal C_k)\}\);
current branch removed\\
\(\mathsf{RestoreLive}(b,k,\sigma)\) &
\(\mathcal E[b{:}\mathcal C],\ k\in\operatorname{dom}K\) &
\(\mathcal E[b{:}\mathsf{carry}(\mathcal C)
 \parallel_g b_k{:}\mathsf{clone}_k(\mathcal C_k)]\) &
\(\{\mathsf{Keep}(\mathcal C),\mathsf{Copy}_k(\mathcal C_k)\}\)\\
\(\mathsf{MergeSelect}(g,w,\sigma)\) &
\(\mathcal E[T_L\oplus_g^sT_R],\
 s\in\{\mathsf{open},w\}\) &
\(\mathcal E[\mathsf{carry}(T_w);b_J{:}E_\sigma^J]\) &
\(\{\mathsf{Keep}(T_w)\}\);
other arm removed\\
\(\mathsf{MergeJoin}(g,\sigma)\) &
\(\mathcal E[T_L\parallel_gT_R]\) &
\(\mathcal E[\mathsf{join}(\mathsf{carry}(T_L),
 \mathsf{carry}(T_R));b_J{:}E_\sigma^J]\) &
\(\{\mathsf{Keep}(T_L),\mathsf{Keep}(T_R)\}\)\\
\bottomrule
\end{tabular}
\end{table*}

In an open choice, the first processed workflow call atomically selects its arm and disables later calls from the other arm, whether that first call creates or reuses an authorization record.
Consequently, \(\mathsf{MergeSelect}(g,w,\sigma)\) is valid precisely while the unselected arm has made no recorded progress.
\(\mathsf{MergeJoin}\) replaces the current \(\parallel_g\) node with \(\mathsf{join}\), removes \(g\) from the set of IDs an edit may name, and prevents the same pair from being joined again.

\subsubsection{Preparing calls for the new rules}

Every Fork, Restore, or Merge edit prepares replacements for all runtime rules in its policy domain.
Before the edit takes effect, \(\beta(x)\) uses \(\eta^-=\eta\) for every current call \(x\).
The judgment for these six forms consists exactly of one row of \cref{tab:history-rules} applied inside the well-typed workflow context \(\mathcal E\), followed by the state update below.
Let \(e_u\) contain the edit record and exactly the copied, added, and group nodes prescribed by \(T'\); carried node IDs remain unchanged.
For every current target workflow call \(x\), the trusted runtime deterministically creates a rule entry \(j_x^+\) and authorization token \(h_x^+\).
The updated record \(\rho'\) preserves the action, source region, scope, digest, and retry information.
Define
\[
\begin{aligned}
 \mathcal H_u^+
   &=\{(x,j_x^+,h_x^+)\mid x\in X_{\sem{T'}}\},\\
 \beta'(x)&=(j_x^+,\eta^+).
\end{aligned}
\]
The new authorization tokens are inactive and inaccessible to the caller before the rule update takes effect.
The complete judgment is shown below.
\begin{equation}
\label{eq:derived-cut}
\begin{aligned}
 \Sigma_\zeta;H\vdash u\Downarrow
   &(H',\mathcal C_{H,u},\eta^-,\eta^+),\\
 H'&=(\omega,G\mathbin{\cdot}e_u,T',K,\Sigma_\zeta,
      \rho',\beta',\chi,\mathsf{ver}+1,\eta^+),\\
 \mathcal C_{H,u}&=\sem{T'},\qquad \eta^-=\eta .
\end{aligned}
\end{equation}
Here \(H'\) is the derived target record.
It becomes current and activates \(\eta^+\) only through the atomic rule update.
In this record, \(G\mathbin{\cdot}e_u\) retains earlier progress, the authority record for each result, and existing nodes.
It initializes each copied or added branch from its remaining registered contract, gives each result an inherited or rule-signed authority record, and records no progress for the new branch.
Thus the derivation preserves \(K,\Sigma_\zeta,\chi\), appends one record of the edit and its target nodes, and prepares a new rule entry and authorization token for every target call under \(\eta^+\).
Table constructors determine new group modes, while carried nested modes remain unchanged.
A nonexistent checkpoint, wrong group mode, nonmember winner, mismatched rule pattern, unauthorized removal, or undefined composition produces \(\mathsf{Invalid}\).
Independent domains step by asynchronous product, and each domain replaces all of its runtime rules at once.

\begin{lemma}[Exact edit derivation]
\label{lem:history-derivation}
For well-formed \(H\), the judgment in \cref{eq:derived-cut} is deterministic.
If it derives \(H'\), then \(H'\) is well formed.
It preserves required results and their signed authorities, retaining a record for every satisfied result and every still-required source result except a signed authorized removal.
The surrounding-workflow identity map and family \(\mathcal P_u\) preserve which calls refer to the same action, source regions, the authority for each result, and causal order.
For every \(a\in\mathsf{StillRequired}_u\), some complete target result \(t\) satisfies \(\mathsf{Covers}_u(t,a)\), including every unchanged result outside the edited region.
Every current target workflow call \(x\) has \(\beta'(x)=(j_x^+,\eta^+)\) and an authorization token for \(j_x^+\).
\end{lemma}
\begin{proof}[Proof sketch]
The six cases use the copy/keep bijections and surrounding-workflow identity map to preserve action sharing, source regions, and order.
They also construct a target result satisfying each required \(\mathsf{Covers}_u\) relation.
The common state update establishes well-formedness and sets \(\beta'(x)=(j_x^+,\eta^+)\) for every current target workflow call.
\Cref{app:schema-fidelity-proof} gives the full argument.
\end{proof}

This section defines the checker semantics in two steps.
It first classifies each workflow call as a new authorization or reuse and then computes the largest set of executions that remains safe throughout execution.

Recall the authorization-label alphabet \(\mathcal U\) and regular policy \(\mathcal A\subseteq\mathcal U^*\), which is closed under truncation.
For authorization log \(\Delta\), \(\mathsf{lab}(\Delta)\in\mathcal A\) is its authorization sequence and \(D_\Delta\) is the set of actions that already have authorization records.
Fix \(\Theta=(\kappa,H,\Delta,\mathsf{out},\mathcal A,u)\) and write \(\mathsf{state}(\Theta)=(\kappa,H,\Delta,\mathsf{out},\mathcal A)\).
\(\mathsf{Derive}\) is a partial function.
For a valid initial rule update, a Fork, Restore, or Merge edit, or a registered extension, it produces a unique derived record.
Otherwise, it is undefined.
When defined, write \(\mathsf{Derive}(\Theta)=(\kappa_u,H_u,\mathcal A_u,\mathcal C_u,\eta^-,\eta^+)\).
Root derivation additionally requires that \(\mathsf{CheckReg}(\mathcal W,\mathcal R)\) accept.
Each Fork, Restore, or Merge edit has \((\kappa_u,\mathcal A_u)=(\kappa,\mathcal A)\) and uses the unique edit judgment.
\(\mathsf{Extend}(\xi)\) adds signed entries for edit rules, calls, the authority for each result, and policies while preserving the current workflow, its progress, and all earlier authorizations.
\Cref{app:registered-extension} gives the judgment.
Let \(\rho_u\) be the target registry in \(H_u\), let \(O_{H,u}=O_{\mathcal C_u}\), and use \(H'=H_u\) and \(\rho'=\rho_u\) where earlier notation is convenient.
The execution sets below use \(\Theta\), the edit-rule registry, \(\rho_u\), and target policy \(\mathcal A_u\).
Before the rules change, the trusted runtime checks the complete record again, including queued requests \(\mathsf{out}\), so sending a request concurrently makes an earlier answer stale.
We abbreviate the fixed parameters by writing \(H,u\).

\subsection{New and Reused Authorizations}
The same tool action may appear at several workflow-call positions after Fork or Restore, but it must be authorized at most once.
Function \(\Resolve^{\rho_u}_\Delta\) keeps every workflow call and annotates it as creating a new authorization or reusing an earlier one.
Starting with \(D_0=D_\Delta\), it maps a call sequence \(w=x_1\cdots x_n\) to annotated sequence \(z_1\cdots z_n\).
For \(d_i=q^{\rho_u}_{\rm act}(x_i)\),
\begin{equation}
\label{eq:resolve}
 z_i=
 \begin{cases}
  \Fresh(x_i,d_i), & d_i\notin D_{i-1},
       \quad D_i=D_{i-1}\cup\{d_i\};\\
  \Alias(x_i,d_i), & d_i\in D_{i-1},
       \quad D_i=D_{i-1}.
 \end{cases}
\end{equation}
\(\mathsf{auth}\) keeps the authorization label of each new authorization and contributes nothing for reuse.
\[
 \mathsf{auth}(\Fresh(x,d))=\ell(d),\qquad
 \mathsf{auth}(\Alias(x,d))=\epsilon .
\]
For annotated \(z\), \(\mathsf{raw}(z)\) drops the new/reuse annotation and action while preserving the ordered workflow-call IDs.
A \Retry repeats the same canonical request and adds no workflow call.
Restored workflow calls that share an action remain distinct.
The first call creates the authorization record, and later calls reuse it.

\begin{lemma}[New/reuse classification]
\label{lem:resolution}
\(\Resolve^{\rho_u}_\Delta\) is deterministic, preserves length, and gives the same annotations for the calls shared at the start of two sequences.
Among calls annotated as new, each action occurs at most once outside \(D_\Delta\), and dropping the annotations gives exactly the input sequence.
Moreover, if \(\Resolve^{\rho_u}_\Delta(rv)=zv'\), then \(z=\Resolve^{\rho_u}_\Delta(r)\) and \(v'=\Resolve^{\rho_u}_{\Delta_z}(v)\), where \(\Delta_z\) extends \(\Delta\) by the new authorizations in \(z\).
\end{lemma}
\begin{proof}
Induction on input length shows that each annotation is determined by the calls seen so far and that \(D_i\) changes exactly when the function creates a new authorization.
Set membership ensures at most one new authorization per action.
Both cases retain \(x_i\).
Splitting after \(r\) gives the stated composition equation.
\end{proof}
Thus a retry changes neither the processed-call record nor the authorization sequence, although its successful API reply is observable as \(\mathsf{Return}(t,[v]_{\cong})\).
A copied workflow call that reuses an authorization record advances the processed-call record without changing the authorization sequence.

\subsection{Safe Completion Throughout Execution}

\subsubsection{Computing the largest safe execution set}

For each derived result \(o\in O_{H,u}\), define all of its complete executions and the subset allowed by policy.
\begin{equation}
\label{eq:safe-completions}
\begin{aligned}
 R_o(H,u)&=\{\Resolve^{\rho_u}_\Delta(w)\mid w\in\Lin(p_o)\},\\
 M_o(H,u)&=\{z\in R_o(H,u)\mid
   \mathsf{lab}(\Delta)\mathsf{auth}(z)\in\mathcal A_u\}.
\end{aligned}
\end{equation}
Here an execution is complete when every workflow call for its result has been processed, and requests and returns add no \Fresh or \Alias events.
For an annotated partial execution \(z\), let
\[
 \mathsf{Compat}(z)=
 \{o\mid \mathsf{raw}(z)\preceq w
		          \text{ for some }w\in\Lin(p_o)\}.
\]
For a set \(B\) of finite sequences, write
\(\Partial B=\{z\mid \exists y\in B.\ z\preceq y\}\).
\(\mathsf{PreservesRequired}(H,u,\widehat B)\) means that every result in \(\mathsf{StillRequired}_u\) is covered by some target result represented in \(\widehat B\):
\[
 \forall a\in\mathsf{StillRequired}_u.\quad
 \exists (t,y)\in\widehat B.\ \mathsf{Covers}_u(t,a).
\]

\begin{definition}[Safe execution set]
\label{def:safe-execution-set}
The checker tracks both the partial executions that the runtime may currently allow and the complete executions that justify allowing them.
A safe execution set is a pair \((L,\widehat B)\), where \(L\) contains allowed partial executions and \(\widehat B\) contains allowed complete executions labeled by result, satisfying the conditions below.
At the laptop Fork, for example, \(\widehat B\) contains the \(L\) and \(R\) payment--shipment executions labeled by \(o_L\) and \(o_R\), while \(L\) contains their partial executions.
\begin{enumerate}[leftmargin=*,itemsep=1pt]
 \item \(\widehat B\ne\varnothing\);
 \item \emph{workflow and results:} \(\mathsf{PreservesRequired}(H,u,\widehat B)\), \(\widehat B\subseteq \biguplus_{o\in O_{H,u}}\{o\}\times R_o(H,u)\), and \(L\subseteq\Partial(\bigcup_{o\in O_{H,u}}R_o(H,u))\);
 \item \emph{policy:} \(\mathsf{lab}(\Delta)\,\mathsf{auth}(z)\in\mathcal A_u\) for every \(z\in L\);
 \item \emph{completion from every partial execution:} \(L=\Partial(\pi_2\widehat B)\); and
 \item \emph{every compatible result remains completable:} for every \(z\in L\) and \(o\in\mathsf{Compat}(z)\), some \((o,y)\in\widehat B\) extends \(z\).
\end{enumerate}
\end{definition}

Checking each result separately is insufficient.
After a partial execution shared by several results, every result still compatible with that partial execution must retain a policy-safe completion.
The operator below removes exactly the complete executions that violate this condition.
Removing one completion may make another partial execution unsafe, so the checker repeats this pruning until no further execution must be removed.
Put
\[
 \widehat B_0=\biguplus_{o\in O_{H,u}}\{o\}\times M_o .
\]
For \(\widehat B\subseteq\widehat B_0\), define the monotone operator
\begin{equation}
\label{eq:robust-operator}
\widehat\Phi(\widehat B)=
\left\{(o,y)\in\widehat B_0\ \middle|\
\begin{array}{l}
\forall z\preceq y.\ \forall o'\in\mathsf{Compat}(z).\\[-2pt]
\exists (o',y')\in\widehat B.\ z\preceq y'
\end{array}\right\}.
\end{equation}
Set \(\widehat B_{i+1}=\widehat\Phi(\widehat B_i)\).
Since \(\widehat B_1\subseteq\widehat B_0\), repeated application removes executions until the finite set stops changing:
\[
\begin{aligned}
 \widehat B^\dagger_{H,u}
   &=\nu\widehat B.\widehat\Phi(\widehat B)
     =\bigcap_{i\geq 0}\widehat B_i,\\
 B^\dagger_{H,u}&=\pi_2\widehat B^\dagger_{H,u},
 &W^\dagger_{H,u}&=\Partial(B^\dagger_{H,u}).
\end{aligned}
\]
\begin{definition}[Safe Fork, Restore, or Merge edit]
\label{def:safe-execution-edit}
\begin{equation}
\label{eq:largest-safe-set}
\mathsf{SafeEdit}(\Theta)
 \quad\Longleftrightarrow\quad
 \begin{cases}
 \widehat B^\dagger_{H,u}\ne\varnothing,
   & \mathsf{Derive}(\Theta)\text{ is defined};\\
 \mathsf{false},&\text{otherwise}.
 \end{cases}
\end{equation}
\end{definition}
At \(z=\epsilon\), the fixed-point condition retains a completion for every target result, and the checked coverage maps cover every required source result.
The existential quantifier chooses a call order, and the universal condition keeps every compatible result completable after every partial execution.

\begin{lemma}[Largest safe execution set]
\label{lem:largest-safe-family}
A safe execution set exists iff \(\mathsf{SafeEdit}(\Theta)\).
When the edit is safe, \((W^\dagger_{H,u},\widehat B^\dagger_{H,u})\) is a safe execution set, and every safe execution set \((L,\widehat B)\) satisfies \(\widehat B\subseteq\widehat B^\dagger_{H,u}\) and \(L\subseteq W^\dagger_{H,u}\).
\end{lemma}
\begin{proof}
The workflow, result, and policy conditions place every safe execution set's complete executions inside \(\widehat B_0\).
The final condition makes this set post-fixed, so monotonicity and finite Knaster--Tarski place it inside \(\widehat B^\dagger\), and the projected partial executions give \(L\subseteq W^\dagger\).
Conversely, \(\widehat B^\dagger=\widehat\Phi(\widehat B^\dagger)\) keeps every compatible result completable.
At \(z=\epsilon\), this condition and the checked coverage witnesses from \cref{lem:history-derivation} establish \(\mathsf{PreservesRequired}\).
The inclusion \(\widehat B^\dagger\subseteq\widehat B_0\) ensures that complete executions follow the edited workflow and policy, while closure of \(\mathcal A_u\) under truncation makes every allowed partial execution safe.
Finally, \(W^\dagger=\Partial(\pi_2\widehat B^\dagger)\) gives a completion after every allowed partial execution.
Hence the fixed point is a safe execution set exactly when nonempty.
\end{proof}

\(\mathsf{MaxSafe}(\Theta;L,\widehat B)\) means that \((L,\widehat B)\) is safe and contains every other safe execution set.
When such a pair exists, write its unique value as \(\mathsf{SafeBehavior}(\Theta)\).
\Cref{lem:largest-safe-family} proves that this witness exists exactly when \(\widehat B^\dagger_{H,u}\ne\varnothing\) and then equals \((W^\dagger_{H,u},\widehat B^\dagger_{H,u})\).

\subsubsection{Why Result-by-Result Checking Fails}

\Cref{sec:implications} introduced the contract \(X\otimes(Y\oplus Z)\) under policy \(\{\epsilon,x,y,xz,yx\}\).
For its annotated calls, write \(\bar x=\Fresh(x_o,d_x)\), \(\bar y=\Fresh(y_o,d_y)\), and \(\bar z=\Fresh(z_o,d_z)\).
The candidate family initially contains \(\bar y\bar x\) and \(\bar x\bar z\).
The first fixed-point iteration removes \(\bar x\bar z\).
The second iteration removes \(\bar y\bar x\), because after the first removal \(z=\epsilon\) no longer has a completion for the \(X\otimes Z\) result.
The empty fixed point is the formal reason the exact checker rejects.

\begin{proposition}[Exact correspondence with generalized nonblocking]
\label{prop:result-necessary}
For nonempty \(\widehat B\subseteq\widehat B_0\), let \(\mathcal T_{\widehat B}\) be the trie for \(L_{\widehat B}=\Partial(\pi_2\widehat B)\).
Mark state \(z\) by \(\mathsf{possible}_o\) exactly when \(o\in\mathsf{Compat}(z)\), and mark state \(y\) by \(\mathsf{done}_o\) exactly when \((o,y)\in\widehat B\).
Then
\[
\begin{gathered}
 \widehat B\subseteq\widehat\Phi(\widehat B)\\[-2pt]
 \Longleftrightarrow\
 \forall o\in O_{H,u}.\;
 \mathcal T_{\widehat B}\text{ is }
 (\mathsf{possible}_o,\mathsf{done}_o)\text{-nonblocking}.
\end{gathered}
\]
Consequently, if \(\widehat B^\dagger\) is nonempty, it is the largest nonempty completion subfamily of \(\widehat B_0\) satisfying the nonblocking condition for every result.
If it is empty, no nonempty subfamily satisfies the conditions.
Ordinary marker nonblocking of the unaugmented projected language is strictly weaker.
\end{proposition}
\begin{proof}
From reachable partial execution \(z\), a \(\mathsf{done}_o\)-marked state is reachable exactly when some \((o,y)\in\widehat B\) extends \(z\).
Quantifying over exactly the states marked \(\mathsf{possible}_o\) is therefore the post-fixed-point condition.
Greatestness follows from the greatest-fixed-point construction.
For strictness, the preceding instance has the marker-nonblocking projected pair \((\Partial\{\bar y\bar x,\bar x\bar z\}, \{\bar y\bar x,\bar x\bar z\})\), but its indexed fixed point is empty.
\end{proof}

\Cref{prop:result-necessary} relates the fixed point to generalized nonblocking after the trusted runtime has derived the completion condition for each result~\cite{dequeiroz2005multitasking,malik2008generalised}.
The trusted runtime derives those conditions from the requested edit, authorized removals, actions, authorization records, and the execution record.
It also decides whether each tool call creates or reuses an authorization, computes the largest safe execution set or a rejection proof, and puts an accepted edit into effect atomically.

\subsubsection{Rejection Proofs}

Each removed \((o_{\rm rem},y)\) records its rank \(i\), a partial execution \(z\preceq y\), and uncovered \(o_{\rm miss}\in\mathsf{Compat}(z)\).
Induction shows that every post-fixed \(P\) is contained in every \(\widehat B_i\).
An empty final family therefore proves that no safe execution set exists.
Before the rule update, the trusted runtime recomputes both the removal order and the remaining completions.

\subsubsection{Checking Again After a Permitted Partial Execution}

To show that the checker can be applied again after each permitted partial execution, fix \(H,u,\mathcal A_u\) after deriving \(\mathcal C\), and write \(W^\dagger(\mathcal C,\Delta)\) and \(\widehat B^\dagger(\mathcal C,\Delta)\) with those parameters suppressed.
\begin{lemma}[Rechecking after a partial execution]
\label{lem:recheck-coherence}
If \(z\in W^\dagger(\mathcal C,\Delta)\), then
\(\Resolve^{\rho'}_\Delta(\mathsf{raw}(z))=z\),
\(\mathsf{lab}(\Delta_z)=
 \mathsf{lab}(\Delta)\mathsf{auth}(z)\), and
\[
\begin{aligned}
 W^\dagger(\mathcal C,\Delta)/z
  &=W^\dagger(\mathcal C/\mathsf{raw}(z),\Delta_z),\\
 \widehat B^\dagger(\mathcal C,\Delta)/z
  &=\widehat B^\dagger(
      \mathcal C/\mathsf{raw}(z),\Delta_z),
\end{aligned}
\]
where
\(\widehat B/z=\{(o,v)\mid(o,zv)\in\widehat B\}\).
\end{lemma}
\begin{proof}[Proof sketch]
Applying \(\Resolve\) incrementally recovers the original calls, gives the authorization-log equality, and makes the candidate executions after \(z\) identical on both sides.
Removing the executed sequence \(z\) from \(\widehat B^\dagger\) gives a post-fixed family for the next check, while adding \(z\) to any such family gives one for the original check.
Greatestness yields the fixed-point equality, and \(\Partial(B)/z=\Partial(B/z)\) yields the generated-language equality.
The full proof is in \cref{app:recheck-coherence-proof}.
\end{proof}
Thus every accepted partial execution retains a safe completion for every result that remains compatible with that execution.

\subsubsection{Checker Outputs}

The checker returns \(\mathsf{Invalid}\) when derivation fails, \(\mathsf{Reject}\) with the ordered reasons for removing every execution when the fixed point is empty, and \(\mathsf{Accepted}\) with the largest safe execution set otherwise.
Write \(\mathsf{Check}(\Theta)\) for this deterministic output.
Each answer includes a digest of the execution record that was checked.
\Cref{app:returned-data} gives the exact returned data and verification rules.
\Cref{prop:app-output-complexity} proves an output-sensitive \(O(D+n\Lambda+|\mathcal G_{\rm cov}|+V_B\log(2+V_B))\)-time, \(O(D+\Lambda+|\mathcal G_{\rm cov}|)\)-space bound.
Exponential growth occurs only in explicitly emitted linearizations and pairs linking partial executions to required results.

\section{Runtime Enforcement}
\label{sec:enforcement}

The checker produces either a rejection proof or the largest safe execution set.
The trusted runtime turns a nonempty safe execution set into rules checked before each tool call.
It builds a finite automaton for the allowed partial executions, records which version of the execution record was checked, and atomically puts the new rules into effect.
The protocol orders every overlapping tool call either before or after the atomic rule update.
The proof compares this runtime protocol with an ideal machine that puts an accepted edit into effect in one atomic step and allows exactly its safe partial executions.
The ideal transition system and the detailed runtime rules appear in \cref{sec:ideal,tab:app-event-matrix}.
This section gives the automaton construction, the invariant after the rules change, and the atomic update used by the formal guarantees in the next section.

\subsection{Turning Safe Executions into Runtime Rules}

The runtime enforces an accepted safe execution set by tracking exactly which safe partial executions remain possible after each processed call.
Removing result labels from the largest safe execution set gives the finite pair \((W,B)=(W^\dagger_{H,u},B^\dagger_{H,u})\).
For \(L/z=\{v\mid zv\in L\}\), let \(q_z=(W/z,B/z)\) be the automaton state after executing \(z\).
\begin{equation}
\label{eq:canonical-automaton}
\begin{aligned}
 Q_{W,B} &= \{q_z\mid z\in W\},\\
 q^0&=q_\epsilon=(W,B),\\
 q_z\xrightarrow{a}q_{za}
 &\quad\Longleftrightarrow\quad za\in W .
\end{aligned}
\end{equation}
Call this deterministic automaton \(\mathsf{Can}(W,B)\).
For the laptop example, its initial state enables either supplier's payment; after the \(L\)-payment transition, it enables the \(L\)-shipment transition but not the \(R\)-payment transition.
Every state is reachable and is marked exactly when \(\epsilon\in B^\dagger/z\).
Thus a marked state represents a completed execution, while an unmarked state records safe work that may still continue.

Each transition symbol \(a\) contains workflow call \(x\), action \(d\), and its \Fresh or \Alias mode, so one workflow call cannot use another call's rule entry.

\begin{lemma}[Exact finite enforcement]
\label{lem:runtime-language}
The generated and marked languages of the automaton in \cref{eq:canonical-automaton} are exactly \((W^\dagger_{H,u},B^\dagger_{H,u})\).
The automaton is the smallest deterministic marked partial automaton for this pair up to state renaming.
\end{lemma}
\begin{proof}
Induction on \(z\) shows that the automaton reaches \((W^\dagger/z,B^\dagger/z)\), enables \(a\) iff \(za\in W^\dagger\), and marks the reached state iff \(z\in B^\dagger\).
Distinct states differ on a remaining generated or marked sequence and must remain distinct.
The automaton is the paired-language derivative construction~\cite{brzozowski1964derivatives}.
\end{proof}

\subsection{Atomic Rule Update}

Before the rules change, the trusted runtime re-derives the edit and recomputes the fixed point from the execution record.
If any tool call has committed since the candidate was computed, the candidate is stale and changes nothing.
Otherwise, one policy-domain transaction closes the old rule version, activates the new automaton, replaces the rule entry and authorization token for every current workflow call, and only then returns success to the Agent.
For each tool call, a transaction verifies its workflow-call ID, action, canonical request, authorization token, and enabled automaton transition before recording progress.
A \Fresh call creates one authorization record and queues the canonical request.
An \Alias call advances the workflow by linking to the existing record without extending the authorization sequence.

\paragraph{Invariant after the rule update}
We write \(\mathsf{AgentSec}(S;c,r)\) when five conditions hold: (1) the workflow now in force and its required results come from the checked edit; (2) the stored nonempty execution set is the checker's largest safe execution set; (3) the workflow, authorization log, queue, and checkpoints match the processed-call sequence \(r\); (4) the automaton is in exactly the state reached after \(r\); and (5) exactly one rule version is active, with one rule entry and authorization token for every remaining workflow call.
For a submitted tool request, one transaction verifies the workflow call, action, canonical request, authorization token, active rule version, and enabled automaton transition before recording any progress.
The same transaction either creates the authorization record and queues the request, or records reuse of an existing authorization record and its return value.

\paragraph{Final check}
For current state \(S\), write
\(\Theta_S(u)=(\kappa,H,\Delta,\mathsf{out},\mathcal A,u)\)
for the safety-check instance obtained from its current components.
Write \(\mathsf{Ready}(S,u,Y)\) when the trusted runtime re-derives \(u\) from the current execution record, recomputes the same nonempty fixed point and automaton carried by \(Y\), confirms that the record has not changed since \(Y\) was computed, confirms that the old rule version is current, and confirms that the new version is unallocated or inactive.
If these checks pass, one transaction closes the old version, activates the new automaton, replaces every current rule entry and authorization token, and then returns success to the Agent.
A changed state makes the candidate stale and forces a new check.
The complete state updates, retries, request sending, completion updates, and restart rules appear in \cref{app:runtime-transitions,tab:app-event-matrix}.

\begin{lemma}[Edit--call serialization]
\label{lem:cut-linearization}
In every reachable well-formed runtime state, a visible edit answer and a tool call in the same policy domain have a unique order at the atomic rule update.
If the call commits first, it changes the execution record, so the earlier answer becomes stale and the edit cannot take effect.
If the rule update commits first, the old call cannot resolve \Fresh or \Alias.
Steps in independent domains commute by product semantics.
\end{lemma}
\begin{proof}
If the call commits first, its recorded progress changes the execution record.
If the rule update commits first, it closes the old rule version, so the old authorization token fails.
Atomicity excludes a partially applied third case.
Independent domains commute.
\end{proof}

\subsection{Trust Boundary}

The TCB is the trusted boundary defined in \cref{sec:model}.
Search, minimization, and candidate preparation remain untrusted because the final check recomputes them.

\section{Formal Guarantees}
\label{sec:results}

The formal results connect derivation from the execution record, exact checking, finite enforcement, the atomic rule update, and every later runtime step.
Their central exactness result says that the checker accepts if and only if a safe runtime implementation exists.
The remaining results establish faithful registration, the information required for an exact answer, repeated use, and runtime correctness.

When defined, write
\(\delta=\mathsf{Derive}(\Theta)
 =(\kappa_u,H_u,\mathcal A_u,\mathcal C_u,\eta^-,\eta^+)\).
A \emph{safe runtime implementation} is an automaton \(M\) together with complete executions \(\widehat B_M\) labeled by result.
The partial and complete executions of \(M\) must form a safe execution set as defined in \cref{def:safe-execution-set}.
The atomic rule update activates this automaton at \(\eta^+\), closes \(\eta^-\), and replaces the rule entries and authorization tokens for all target calls.

\paragraph{What the checker must know}
Write \(V=\mathsf{Sec}(S;c,r)
=({\mathbf P},{\mathbf I},{\mathbf E},{\mathbf C})\) for four parts of the checked record: (1) the workflow and its required results, (2) which calls refer to the same action, (3) earlier authorizations and call progress, and (4) which rules govern a call that overlaps the atomic rule update.
Suppose a checker receives only some function \(f(V)\) of this record.
The function is exact if any two records that it makes indistinguishable always give the same answer to every common decision or rule-update request.
Write \(\mathsf{Exact}(f)\) for this property.
For example, a view that erases the active rule version can conflate a prepared update that may commit with the same update in a state where it must return \(\mathsf{NoCommit}\).
Internal identifiers may be renamed consistently because their spelling cannot affect the answer.
For \(J\subseteq\{\mathbf P,\mathbf I,\mathbf E,\mathbf C\}\), \(\mathsf{Keep}_J\) gives the checker exactly the parts named by \(J\).
\Cref{app:record-information} gives one record pair for each part and a fifth pair for the correspondence between earlier authorizations and actions.

\paragraph{Preserving required results}
\[
\begin{aligned}
\mathsf{PreservesRequired}(H,u,\widehat B)
&\Longleftrightarrow
\forall a\in\mathsf{StillRequired}_u.\\[-2pt]
&\qquad\exists(t,y)\in\widehat B.\
\mathsf{Covers}_u(t,a).
\end{aligned}
\]
Thus every required source result appears in an accepted target result.

\paragraph{Visible runtime behavior}
\(\mathcal O_{\rm api}\) includes Checkpoint decisions, edit answers tied to the checked record, edits put into effect, tool-call answers, sent requests, completion updates, and returned values.
It hides internal computation, stale candidates, and restart bookkeeping.
\[
\mathsf{obs}_{\rm api}(\ell)=
\begin{cases}
\ell,&\ell\in\mathcal O_{\rm api};\\
\tau,&\text{otherwise}.
\end{cases}
\]
Write \(S\approx_{\mathcal O_{\rm api}}^{\mathsf{di}}I\) when \(S\) and \(I\) are related by a divergence-insensitive weak bisimulation under \(\mathsf{obs}_{\rm api}\).
\(\mathcal B\) matches ideal and runtime states on \(\mathsf{AgentSec}\) and these visible fields.
\Cref{app:bisimulation} gives the full relation.

\begin{theorem}[Faithful Registration]
\label{thm:faithful-registration}
For finite-state \(\mathcal W\) and finite \(\mathcal R\),
\(\mathsf{CheckReg}(\mathcal W,\mathcal R)\) accepts iff the stored \(\lambda\) witnesses \(\mathcal W\sqsubseteq_\lambda\mathcal R\).
On acceptance, the source workflow's tool-call traces and \(\mathsf{BRuns}(\mathcal R)\) are order-isomorphic and preserve results, actions, \Fresh/\Alias, the safety answer, and the largest safe execution set.
Malformed registrations fail before startup.
\end{theorem}

For the remaining results, assume accepted registration, trusted checking of every tool call, and linearizable transactions that remain atomic after a restart within each policy domain.

\begin{theorem}[Exact Checkpoint]
\label{thm:exact-checkpoint}
The ideal and runtime machines make the same decision for \(\mathsf{Checkpoint}(k,b)\).
They accept exactly when \(k\) is fresh, \(b\) is current, and the stored work matches the branch's remaining registered work: \(\mathcal C_b=\Gamma_b/\mathsf{raw}(\chi_b)\).
Acceptance preserves \(\mathsf{AgentSec}\) without changing the current workflow or any earlier authorization.
\end{theorem}

Let \(\mathsf{Ext}=\{\mathsf{Extend}(\xi)\mid \xi\text{ is verified}\}\).
The theorems below hold for every \(u\in\mathsf{Edit}_6\uplus\mathsf{Ext}\) at every finite reachable \(\mathsf{AgentSec}\) state.

\begin{theorem}[Exact Safety Checking]
\label{thm:exact-checking}
The ideal and runtime machines return the same mutually exclusive and exhaustive answer for \(\Theta=\Theta_S(u)\).
Undefined derivation gives state-preserving \(\mathsf{Invalid}\).
For defined \(\delta=\mathsf{Derive}(\Theta)\), \(\widehat B^\dagger_{H,u}=\varnothing\) gives state-preserving \(\mathsf{Reject}\) with no safe runtime implementation.
Otherwise, for \(Y=\mathsf{Check}(\Theta)\), nonempty \(\widehat B^\dagger_{H,u}\), existence of a safe runtime implementation, \(\mathsf{Ready}(S,u,Y)\), and a successor that puts the edit into effect while preserving \(\mathsf{AgentSec}\) are equivalent.
The complete-execution component of the largest safe execution set satisfies \(\mathsf{PreservesRequired}\), and the derived record identifies every already-satisfied source result and every authorized removal.
\end{theorem}

\begin{theorem}[Necessary Record Information]
\label{thm:necessary-state}
The full record is sufficient.
If any one of the four parts is omitted wholesale while the others are retained, two records become indistinguishable even though the correct answer differs.
\[
\mathsf{Exact}(\mathsf{Keep}_J)\Longleftrightarrow
J=\{\mathbf P,\mathbf I,\mathbf E,\mathbf C\}.
\]
The four component-labeled rows in \cref{tab:record-pairs} witness omission of \(\mathbf P,\mathbf I,\mathbf E,\mathbf C\).
The \(\mathsf{Drop}_{ra}\) row shows that retaining authorization records and actions without recording which ones correspond is also insufficient.
In particular, a checker is not exact if it forgets which calls or authorization records refer to the same action, or if it sees only the caller's description of what should happen next.
Generalized nonblocking can serve as an exact solver only after the trusted runtime derives the workflow and the completion condition for each result.
\end{theorem}

\begin{theorem}[Repeated-Use Safety]
\label{thm:repeated-use}
Every enabled event preserves \(\mathsf{AgentSec}\).
Therefore, the guarantees hold after every finite reachable sequence of edits, registered extensions, tool calls, authorization updates, rule updates, stops, and restarts.
\end{theorem}

\begin{theorem}[Runtime Correctness]
\label{thm:runtime-correctness}
\(\mathcal B\) is a divergence-insensitive weak bisimulation under \(\mathsf{obs}_{\rm api}\), so every pair related by \(\mathcal B\) satisfies \(S_0\approx_{\mathcal O_{\rm api}}^{\mathsf{di}}I_0\).
The observation records Checkpoint decisions, edit answers tied to the checked record, new authorization data, and returned values.
\end{theorem}

\paragraph{Proof roadmap}
Registration reflection follows by enumerating the finite tool-call traces and checking their order- and identity-preserving correspondence.
The matching Checkpoint transitions prove exact Checkpoint.
The edit-rule checks and greatest-fixed-point argument prove exact checking.
Every safe execution set is contained in the computed family, and a nonempty computed family yields a safe runtime implementation.
Five explicit pairs of execution records prove the information result.
Induction over the runtime rules proves repeated use.
A case analysis over edits, tool calls, races with the rule update, returns, and restarts establishes the weak bisimulation.
The appendix gives the complete case proofs.

\section{Mechanization and Executable Validation}
\label{sec:validation}

Six Lean modules pass \texttt{--trust=0} and mechanize the finite linear-contract core used by the paper.
The mechanization proves that the executable checker accepts exactly when a valid safe execution set exists.
It proves that registration checks every tool call and preserves the exact answer for the current execution record.
It also proves that executable workflow editing is sound and complete for all six edit forms.
An accepted edit produces an atomic rule update, and every finite sequence of later calls preserves Lean's \path{AgentSec}.
Separate proofs cover both orders between a tool call and that update, as well as preparation before it takes effect.
The appendix proves the general pomset lift, the maps that preserve source results, the information lower bound, and ideal--runtime weak bisimulation in \cref{app:spec-bridge,app:schema-fidelity-proof,app:record-information,app:bisimulation}.

All 19 paper-specific tests pass, together with 62 authorization and checker regression tests for supporting components.
Across accepted instances with 2--128 results, the checker takes \(0.11\)--\(53.47\) ms.
Rejected instances take \(0.11\)--\(5.51\) ms, while enumerating 6--720 unordered completions takes \(0.33\)--\(50.35\) ms.
\Cref{app:artifact-bounds} reports the complete measurement protocol and executable coverage.

\section{Related Work}
\label{sec:related}

\paragraph{Recovery, workflow update, and synthesis}
Output-commit recovery reconstructs a fixed computation.
Generalized nonblocking and Live Synthesis solve supplied transition systems and completion conditions, while workflow inheritance and dynamic controller update connect supplied old and new specifications~\cite{strom1985optimistic,elnozahy2002rollback,dequeiroz2005multitasking,malik2008generalised,finkbeiner2022live,vanderaalst2002inheritance,nahabedian2020dynamic,amram2022dynamic}.
Our runtime derives the edited workflow, action reuse, still-required results, completion conditions, and new runtime rules directly from the execution record.
It computes the largest safe execution set and atomically makes rules take effect for new authorizations and reuse of existing ones.

\paragraph{Agent recovery mechanisms}
ACRFence blocks replay after Restore, DART selects checkpoints, Atomix delays external actions until commit, and Rebound provides atomic, auditable rollback~\cite{zheng2026acrfence,yang2026dart,mohammadi2026atomix,burke2026rebound}.
Our formal results give exact criteria for Checkpoint and all six forms of Fork, Restore, and Merge, together with a rejection proof and an information lower bound showing that omitting any record part wholesale can change the correct answer.
These guarantees continue through the atomic rule update and arbitrary finite execution.

\section{Conclusion}
\label{sec:conclusion}

Safe execution editing requires the runtime to derive from the execution record what each edit must preserve.
Earlier authorizations remain fixed, and every still-required result retains a policy-safe completion.
For Checkpoint, the six forms of Fork, Restore, and Merge, and registered extensions, our checker either applies the direct Checkpoint check, computes the largest safe execution set, or returns a finite proof that no runtime rule set can enforce the edit safely.
The formal results connect this exact decision to decision-relevant distinctions from all four parts of the execution record, the atomic rule update, repeated execution edits, and equivalence with an ideal atomic machine.

\clearpage
\appendices
\section{Supplementary Proof Appendix}
\label{app:proofs}

This appendix gives the proof details for \cref{thm:exact-checking} over finite registered call models and atomic policy-domain rule versions.
Independent domains retain the product interpretation stated in the paper.

\subsection{The Declarative Definition and the Computed Answer}
\label{app:spec-bridge}

Fix a well-formed safety-check instance \(\Theta=(\kappa,H,\Delta,\mathsf{out},\mathcal A,u)\) for which the edit derivation yields \(\delta=(\kappa_u,H_u,\mathcal A_u,\mathcal C_u,\eta^-,\eta^+)\).
All sets in this subsection retain result labels, so two equal resolved words belonging to different results remain different elements.
\begin{definition}[Required source results]
\label{def:source-results}
Write \(\mathsf{kind}(u)\) for the constructor of \(u\).
Call either Fork constructor a Fork and either Restore constructor a Restore.
For the matching source row in \cref{tab:history-rules}, define
\[
\begin{array}{c|c}
u&\mathsf{SrcReg}(H,u)\\ \hline
\text{Fork}&\{(L,\mathcal C),(R,\mathcal C)\}\\
\text{Restore}&\{(\mathsf{cur},\mathcal C),(k,\mathcal C_k)\}\\
\mathsf{MergeSelect}(g,w,\sigma)&
 \{(w,\sem{T_w}),(\bar w,\sem{T_{\bar w}})\}\\
\mathsf{MergeJoin}(g,\sigma)&
 \{(L,\sem{T_L}),(R,\sem{T_R})\}
\end{array}
\]
The full set of source results is
\[
\begin{aligned}
\mathsf{RequiredBefore}(H,u)
={}&(\{\mathsf{ctx}\}\times O_{\rm unchanged})\\
&\uplus
\biguplus_{\substack{(r,\mathcal D)\\
\in\mathsf{SrcReg}(H,u)}}(\{r\}\times O_{\mathcal D}).
\end{aligned}
\]
Write \(\mathsf{Done}_H(r,o)\) when the recorded progress for region \(r\) shows that every event in result \(o\)'s pomset has completed.
Only \(\mathsf{RestoreReplace}\) treats completed current results as finished:
\[
\mathsf{DoneCur}_H=
\{(\mathsf{cur},o)\mid
o\in O_{\mathcal C}\land\mathsf{Done}_H(\mathsf{cur},o)\},
\]
\[
\begin{array}{c|c}
u&\mathsf{Satisfied}_{H,u}\\ \hline
\mathsf{RestoreReplace}(b,k,\sigma)&\mathsf{DoneCur}_H\\
\text{otherwise}&\varnothing
\end{array}
\]
Completion observability makes this set either all of \(\{\mathsf{cur}\}\times O_{\mathcal C}\) or empty.
The only sets eligible for signed removal are
\[
\mathsf{CurOut}=\{\mathsf{cur}\}\times O_{\mathcal C},
\qquad
\mathsf{Other}_w=\{\bar w\}\times O_{\sem{T_{\bar w}}},
\]
\[
\begin{array}{c|c}
u&\mathsf{CanRemove}(H,u)\\ \hline
\mathsf{RestoreReplace}(b,k,\sigma)&
 \mathsf{CurOut}\setminus\mathsf{Satisfied}_{H,u}\\
\mathsf{MergeSelect}(g,w,\sigma)&\mathsf{Other}_w\\
\text{otherwise}&\varnothing
\end{array}
\]
\(\mathsf{reqauth}_H(a)\) is the authority identity named by the unique signed authority record found through \(a\)'s source region in \(G\).
When \(a\) comes from a checkpoint, the identity is taken from the stored copy in \(K\).
Context identity, carry, and clone preserve that authority identity, while every source region introduced by an edit rule receives the identity of the authority that signed the rule and policy version.
Thus \(\mathsf{reqauth}_H\) is total and immutable on \(\mathsf{RequiredBefore}(H,u)\), including Fork-created \(L/R\) copies.
For every \(a\in\mathsf{CanRemove}(H,u)\), the removal request must contain a signature by \(\mathsf{reqauth}_H(a)\) over
\[
\begin{aligned}
\mathsf{rmmsg}(H,u)=(
&\omega,\kappa,\zeta,\sigma,H,\\
&\mathsf{CanRemove}(H,u),\\
&\mathsf{origin}(\mathsf{CanRemove}(H,u))).
\end{aligned}
\]
Without all required signatures, the edit has no derivation.
With them, \(\mathsf{AuthorizedRemoval}_{H,u}=\mathsf{CanRemove}(H,u)\).
Write \(\mathsf{Sat}_u=\mathsf{Satisfied}_{H,u}\) and
\(\mathsf{Rm}_u=\mathsf{AuthorizedRemoval}_{H,u}\).
The checker requires
{\small
\[
\begin{aligned}
\mathsf{StillRequired}_u
 &=\mathsf{RequiredBefore}(H,u)
   \setminus(\mathsf{Sat}_u\uplus\mathsf{Rm}_u),\\
\mathsf{RequiredBefore}(H,u)
 &=\mathsf{StillRequired}_u
   \uplus\mathsf{Sat}_u\uplus\mathsf{Rm}_u.
\end{aligned}
\]
}
\end{definition}

\begin{definition}[Preserving required results]
\label{app:spc}
\[
\begin{aligned}
\mathsf{PreservesRequired}(H,u,\widehat B)
&\Longleftrightarrow
\forall a\in\mathsf{StillRequired}_u.\\[-2pt]
&\qquad\exists(t,y)\in\widehat B.\
\mathsf{Covers}_u(t,a).
\end{aligned}
\]
Here \(\mathsf{AuthorizedRemoval}_{H,u}\) contains the typed results of the unselected branch in \(\mathsf{MergeSelect}\) or the current-branch results of an authorized \(\mathsf{RestoreReplace}\).
\(\mathsf{Satisfied}_{H,u}\) contains the completed current-branch results of \(\mathsf{RestoreReplace}\).
Both sets are empty for every other row and for a registered extension.
Thus \(\mathsf{StillRequired}_u=\mathsf{RequiredBefore}(H,u)\setminus(\mathsf{Satisfied}_{H,u}\uplus\mathsf{AuthorizedRemoval}_{H,u})\), so \(\mathsf{PreservesRequired}\) covers exactly the source results that are neither completed nor authorized for removal.
\end{definition}
Let \(\mathfrak R_\Theta\) be the finite set of safe execution sets from \cref{def:safe-execution-set}, ordered componentwise by
\[
 (L_1,\widehat B_1)\sqsubseteq(L_2,\widehat B_2)
 \quad\Longleftrightarrow\quad
 L_1\subseteq L_2\ \land\ \widehat B_1\subseteq\widehat B_2 .
\]

\begin{lemma}[The computed answer is exact]
\label{lem:app-spec-bridge}
The declarative predicate \(\mathsf{MaxSafe}(\Theta;L,\widehat B)\) holds for some pair if and only if \(\widehat B^\dagger_{H,u}\ne\varnothing\).
When such a pair exists, it is unique and equals \((W^\dagger_{H,u},\widehat B^\dagger_{H,u})\).
The declarative object and the computed fixed point are defined independently, and the lemma proves their equality.
\end{lemma}

\begin{proof}
Let \((L,\widehat B)\in\mathfrak R_\Theta\).
The workflow, result, and policy requirements give \(\widehat B\subseteq\widehat B_0\).
For every \((o,y)\in\widehat B\), every partial execution \(z\preceq y\) belongs to \(L=\Partial(\pi_2\widehat B)\).
Persistent result coverage therefore supplies, for every \(o'\in\mathsf{Compat}(z)\), a pair \((o',y')\in\widehat B\) extending \(z\).
Thus \(\widehat B\subseteq\widehat\Phi(\widehat B)\).
By monotonicity on the finite lattice \(2^{\widehat B_0}\), every post-fixed set is contained in the greatest fixed point, so \(\widehat B\subseteq\widehat B^\dagger\).
Taking the projected partial executions yields \(L\subseteq W^\dagger\).

Conversely, suppose \(\widehat B^\dagger\ne\varnothing\).
The fixed-point equality keeps every compatible result completable and, at \(\epsilon\), combines with the checked preservation maps to establish \(\mathsf{PreservesRequired}\).
Inclusion in \(\widehat B_0\) ensures that complete executions follow the target workflow and policy.
Closure of \(\mathcal A_u\) under truncation makes every generated partial execution safe, and \(W^\dagger=\Partial(\pi_2\widehat B^\dagger)\) gives a completion after every partial execution.
Hence \((W^\dagger,\widehat B^\dagger)\in\mathfrak R_\Theta\), and the preceding containment makes it greatest.
If two greatest elements existed, each would contain the other componentwise, so they would be equal.
If \(\widehat B^\dagger=\varnothing\), the first half places every hypothetical safe execution set inside the empty set, contradicting its required nonemptiness.
\end{proof}

\begin{corollary}[Exact checking]
\label{cor:app-exact-boundary}
The ideal edit transition is enabled by the declarative \(\mathsf{SafeBehavior}(\Theta)\) exactly when the checker returns a nonempty fixed point.
On that edge, the ideal language \(L^*\) equals the generated language of \(\mathsf{Can}(W^\dagger,B^\dagger)\).
\end{corollary}

\begin{proof}
The first statement is \cref{lem:app-spec-bridge}.
The second combines that lemma with \cref{lem:runtime-language}.
\end{proof}

\begin{proposition}[Output-sensitive checking]
\label{prop:app-output-complexity}
Fix a unit-cost model for the checker with constant-time finite-map lookup and policy-automaton transitions.
Let
\[
\begin{aligned}
D&=\lvert\Theta\rvert_{\rm chk},&
n&=\max_o\lvert X_{p_o}\rvert,\\
\Lambda&=\sum_o\sum_{w\in\Lin(p_o)}(\lvert w\rvert+1),&
N_B&=\lvert\widehat B_0\rvert,\\
V_B&=\sum_{(o,y)\in\widehat B_0}(\lvert y\rvert+1).
\end{aligned}
\]
\[
\mathcal G_{\rm cov}
=\left\{((o,y),z,o')\ \middle|\
\begin{aligned}
&(o,y)\in\widehat B_0,\ z\preceq y,\\[-2pt]
&o'\in\mathsf{Compat}(z)
\end{aligned}\right\}.
\]
Here \(D\) is the checked input size, \(\Lambda\) is the total size of the emitted linearizations, and \(\mathcal G_{\rm cov}\) contains the triples of candidate execution, partial execution, and required result examined by the checker.
The checker runs in \(O(D+n\Lambda+|\mathcal G_{\rm cov}|+V_B\log(2+V_B))\) time and \(O(D+\Lambda+|\mathcal G_{\rm cov}|)\) space, including the data it returns.
Moreover, \(N_B\leq\sum_o|\Lin(p_o)|\), \(|\mathcal G_{\rm cov}|\leq N_B(n+1)|O_{\mathcal C_u}|\), and the number of nonempty strict-removal ranks is at most \(N_B\).
\end{proposition}

\begin{proof}
A backtracking topological-order enumerator emits all indexed linearizations in \(O(n\Lambda)\) time, while resolution, policy simulation, and construction of raw and resolved tries are linear in emitted volume.
For every support key \((z,o')\), maintain the number of remaining pairs \((o',y')\) extending \(z\), initially enqueue candidates having a required result with zero support, and remove candidates in synchronous batches.
When removal makes a support count zero, enqueue exactly the candidates that depend on that requirement for the next batch.
Each candidate and each triple in \(\mathcal G_{\rm cov}\) is processed at most once, and the batches are exactly \(\widehat B_i\setminus\widehat B_{i+1}\).
Store one rank and cause per removed candidate rather than copied sets, and construct \(\mathsf{Can}(W,B)\) by bottom-up sorting of finite-trie derivative signatures in \(O(V_B\log(2+V_B))\) time.
\end{proof}

\subsection{Ordered Rejection Proofs}
\label{app:rejection-proof}

For an empty fixed point, the returned proof is
\[
C_{\rm fp}=(\widehat B_0,E_0,\ldots,E_k),
\qquad
E_i\subseteq
\widehat B_i\times\Partial(\pi_2\widehat B_i)
\times O_{H,u}.
\]
The verifier first recomputes the canonical \(\widehat B_0(\Theta)\) and requires equality with the serialized first component.
An entry \(((o,y),z,o')\in E_i\) verifies exactly when \(z\preceq y\), \(o'\in\mathsf{Compat}(z)\), and no \((o',y')\in\widehat B_i\) extends \(z\).
The verifier requires the set of first components in \(E_i\) to equal \(\widehat B_i\setminus\widehat\Phi(\widehat B_i)\), sets \(\widehat B_{i+1}=\widehat B_i\setminus\pi_1(E_i)\), and accepts only when \(\widehat B_{k+1}=\varnothing\).
All sets, partial executions, and compatibility checks are finite and use the canonical order fixed by the checker.

\begin{lemma}[The rejection proof is sound and complete]
The verifier accepts \(C_{\rm fp}\) iff the descending fixed-point computation reaches \(\widehat B^\dagger=\varnothing\).
In that case no safe execution set exists.
\end{lemma}
\begin{proof}
The initial equality check ties the proof to \(\Theta\), and every verified round is exactly \(\widehat B_{i+1}=\widehat\Phi(\widehat B_i)\), so canonical iteration supplies completeness.
For soundness, induction places every post-fixed family inside every \(\widehat B_i\), so an empty final family excludes every nonempty safe execution set.
\end{proof}

\subsection{Exact Returned Proof Data}
\label{app:returned-data}

Let \(\mathsf{enc}\) be canonical and length-delimited, with \(\kappa\) authenticating \(\mathcal A\).
The formal transition system treats the domain-separated authenticated digests \(\mathsf{Digest}_{\rm rules}\) and \(\mathsf{Digest}_{\rm record}\) as injective, so equal digests have equal encoded inputs.
For \(\Theta=(\kappa,H,\Delta,\mathsf{out},\mathcal A,u)\), define
\[
\gamma(\Theta)
 =\mathsf{Digest}_{\rm record}(
   \mathsf{enc}(\mathsf{request},\omega,\zeta,\kappa,
   H,\Delta,\mathsf{out},\mathcal A,u)).
\]
When \(\mathsf{Derive}(\Theta)\) is defined and
\(\widehat B^\dagger_{H,u}\ne\varnothing\), also define
\[
\begin{aligned}
\mathsf{aut}_Z
 &=\mathsf{Can}(W^\dagger_{H,u},B^\dagger_{H,u}),\\
H_{\mathsf{aut}}
 &=\mathsf{Digest}_{\rm rules}(
   \mathsf{enc}(\mathsf{aut}_Z,\widehat B^\dagger_{H,u})).
\end{aligned}
\]
The digest checked before the rules change is
\begin{equation}
\label{eq:checked-record-digest}
\begin{split}
 \mathsf{cut}_Z=
 \mathsf{Digest}_{\rm record}(\mathsf{enc}(&\omega,\zeta,\kappa,H,\Delta,
     \mathsf{out},u,\mathcal C_{H,u},\\[-2pt]
     &H_{\mathsf{aut}},\eta^-,\eta^+)).
\end{split}
\end{equation}
Because \(H\) contains \(\chi,\mathsf{ver},\rho,\beta\), the digest covers workflow progress, record version, registry contents, active rule entries and versions, the authorization log, and the request queue.
The verifier re-derives the edit-rule judgment and complete fixed point labeled by result, reconstructs \((W^\dagger,B^\dagger)\), verifies both digests, and checks the full automaton isomorphism
\[
 \mathsf{aut}_Z\cong\mathsf{Can}(W^\dagger,B^\dagger),
\]
including its initial state, marking, and every labeled transition.
Canonical encoding order makes \(\mathsf{Invalid}\), \(\mathsf{Reject}\), and \(\mathsf{Accepted}\) unique and unchanged by consistent renaming of internal identifiers.
A concrete implementation may authenticate \(\mathsf{enc}\) with signatures or collision-resistant digests.
The corresponding guarantee relies on signature unforgeability and digest collision resistance, which the formal model represents with unforgeable identifiers.

\subsection{Proofs for Edit Rules and Rechecking}

\subsubsection{Edit-rule correctness}
\label{app:schema-fidelity-proof}

\begin{proof}[Full proof of \cref{lem:history-derivation}]
Unique object and edit-rule identifiers select one row and record, while context identity preserves everything outside the edited region.
The Fork rows check both copies.
RestoreReplace checks the checkpoint copy and the exact treatment of the current branch, while RestoreLive checks both the retained current branch and the checkpoint copy.
MergeSelect checks the selected and unselected regions, while MergeJoin checks both retained regions before its barrier.
Clone and carry supply the global bijections and preserve which workflow calls belong to each result and which authority controls each source region.
Constructor induction defines \(\mathsf{pr}^u_r\), proves the cover, and supplies every \(\mathsf{Covers}_u\) witness.
The common derived record fixes all remaining fields, preserves or initializes branch progress, and assigns each source region introduced by an edit rule its signed authority record.
Fresh identifiers, partial-constructor checks, and \cref{lem:observable-completion} preserve uniqueness, acyclicity, and contract well-formedness.
Finally, \(\beta'(x)=(j_x^+,\eta^+)\) and \((x,j_x^+,h_x^+)\in\mathcal H_u^+\) for every current target workflow call \(x\).
\end{proof}

\subsubsection{Registered extension}
\label{app:registered-extension}

The Agent may name only an immutable record \(\xi\) that specifies the exact prior edit rules, registry, execution record, and policy, together with finite disjoint additions \(D_\Sigma,D_\rho,D_\alpha\) and a successor policy \(\mathcal A^+\).
The set \(D_\alpha\) contains the signed authority record for every newly introduced result and its source region.
The policy authority and trusted runtime verify \(\xi\).
Existing meanings remain unchanged, new identifiers are fresh and domain-scoped, and \(\mathcal A\subseteq\mathcal A^+\).
Let \(\zeta^+\), \(\kappa^+\), and \(\mathcal A^+\) be the verified successor versions and policy named by \(\xi\).
For every \(x\in X_{\sem T}\), the trusted runtime deterministically creates \(j_x^+\) and \(h_x^+\).
For fresh \(\eta^+\), define
\[
\begin{aligned}
 G^+&=G\uplus D_\alpha,\\
 \Sigma_{\zeta^+}&=\Sigma_\zeta\uplus D_\Sigma,\\
 \rho^+&=\rho\uplus D_\rho,\\
 \mathcal H_\xi^+
   &=\{(x,j_x^+,h_x^+)\mid x\in X_{\sem T}\},\\
 \beta^+(x)&=(j_x^+,\eta^+),\\
 H_\xi^+
   &= (\omega,G^+,T,K,\Sigma_{\zeta^+},\rho^+,\beta^+,\\[-2pt]
   &\hspace{13mm}\chi,\mathsf{ver}+1,\eta^+),\\
 \mathcal C_\xi&=\sem T,\\
 \mathsf{RequiredBefore}(H,\mathsf{Extend}(\xi))
   &=\mathsf{StillRequired}_{\mathsf{Extend}(\xi)},\\
 \mathsf{StillRequired}_{\mathsf{Extend}(\xi)}
   &=\{\mathsf{ctx}\}\times O_{\sem T},\\
 \mathsf{Satisfied}_{H,\mathsf{Extend}(\xi)}&=\varnothing,\\
 \mathsf{AuthorizedRemoval}_{H,\mathsf{Extend}(\xi)}&=\varnothing,\\
 \mathsf{Covers}_{\mathsf{Extend}(\xi)}
   (t,(\mathsf{ctx},o))
   &\Longleftrightarrow t=o .
\end{aligned}
\]
The checked judgment is
\[
(\kappa,H,\mathcal A)\vdash\mathsf{Extend}(\xi)\Downarrow
  (\kappa^+,H_\xi^+,\mathcal A^+,\mathcal C_\xi,\eta,\eta^+).
\]
It preserves \(T,K,\chi,\Delta,\mathsf{out}\) and every existing row of \(G,\Sigma_\zeta,\rho\), appends \(D_\alpha\) to the execution record, and authorizes no removal.
Every current result, workflow call, action, source region, authority record, and causal edge therefore remains unchanged.
It uses the same checker and atomic rule update as a Fork, Restore, or Merge edit.

\begin{lemma}[Registered extension preserves safety]
\label{lem:live-extension}
For a verified \(\mathsf{Extend}(\xi)\), the current safe executions after recorded progress form a post-fixed family for the extension instance and are therefore contained in \(\widehat B^\dagger_{H,\mathsf{Extend}(\xi)}\).
Hence \(\widehat B^\dagger_{H,\mathsf{Extend}(\xi)}\ne\varnothing\).
Consequently, a valid registered extension cannot remove a currently required result or leave it without a safe completion.
\end{lemma}

\subsubsection{Rechecking after a partial execution}
\label{app:recheck-coherence-proof}

\begin{proof}[Full proof of \cref{lem:recheck-coherence}]
Because some projected member of \(\widehat B^\dagger\) has the form \(zv\), the recoverability and composition properties in \cref{lem:resolution} give the recoverability and authorization-log equalities.
Removing the executed sequence from the executions labeled by result retains every compatible result and gives
\[
 \mathsf{Compat}_{\mathcal C/\mathsf{raw}(z)}(s)
 =\mathsf{Compat}_{\mathcal C}(zs).
\]
Together with authorization-sequence concatenation, applying \(\Resolve\) incrementally yields
\[
 \widehat B_0(\mathcal C,\Delta)/z
 =\widehat B_0(\mathcal C/\mathsf{raw}(z),\Delta_z),
\]
with the same retained result identities, policy, and target registry.
Thus \(\widehat B^\dagger/z\) is post-fixed for the next-check operator and lies within its greatest fixed point \(Y\).
Conversely, incremental resolution and policy safety after \(z\) put \(zY=\{(o,zy)\mid(o,y)\in Y\}\) inside the original \(\widehat B_0\).
The union \(\widehat B^\dagger\cup zY\) is post-fixed because partial executions \(t\preceq z\) use existing witnesses, while those of the form \(t=zs\) use witnesses in \(Y\).
Greatestness gives \(zY\subseteq\widehat B^\dagger\), hence \(Y=\widehat B^\dagger/z\).
Finally, \(\Partial(B)/z=\Partial(B/z)\) yields the generated-language equality.
\end{proof}

\begin{proof}[Proof of \cref{lem:live-extension}]
\label{app:live-extension-proof}
The verified extension preserves the current workflow, which calls refer to the same action, every result label, and \(\mathcal A\subseteq\mathcal A^+\).
Every earlier safe completion and coverage witness therefore remains valid for the next fixed-point computation.
The current remaining family is therefore post-fixed, and greatestness gives its inclusion in a nonempty successor fixed point.
\end{proof}

\subsection{Lean Theorem Map}
\label{app:lean-map}

Six Lean modules machine-check the claims used in the paper.
\path{FiniteCore} proves the executable greatest-fixed-point checker equivalent to the declarative safe-execution-set predicate, including resolution over partial executions and greatestness.
\path{RegistrationRefinement} proves that registration accepts the typed finite model, every registered call remains under trusted control, action identity is preserved, and the finite and semantic contracts agree.
\path{HistoryStructure} proves \path{deriveEdit_iff}, well-formedness preservation, and executable fixtures for Choice/Parallel Fork, Replace/Live Restore, Select/Join Merge.
\path{OperationalSemantics} proves invariant preservation across every runtime trace, closure of all six edit forms, registered extension safety, and both orders between a tool call and an atomic rule update.
\path{CompilationBridge} constructs a secure atomic rule update from an accepted checker answer and reflects every successful update back to that answer.
\path{PaperTheorems} combines these results into exact registered checking, six-edit derivation, preservation before rules take effect, trace safety, and serialization between calls and rule updates.
The remaining appendix sections prove the general pomset lift, information lower bounds, and ideal--runtime weak bisimulation.

\subsection{Why No Record Part Can Be Omitted Wholesale}
\label{app:record-information}

We prove that each of the four record parts contains an answer-relevant distinction and that no part can be omitted wholesale while retaining the other three.
For each omitted part, we construct two valid records that become indistinguishable even though their correct answers are different.

Let \(\mathcal Q_{\rm sec}\) pair every valid checked record \(V=\mathsf{Sec}(S;c,r)\) with a well-typed request \(q::=\mathsf{Decide}(u)\mid\mathsf{TryUpdate}(u,Y)\).
For such a record, let \(\Theta_V(u)=\Theta_S(u)\) and \(\mathsf{Ready}_V(u,Y)=\mathsf{Ready}(S,u,Y)\).
For \(\mathsf{Decide}\), \(\mathsf{Ans}\) returns \(\mathsf{Check}(\Theta_V(u))\).
For \(\mathsf{TryUpdate}\), it returns the updated runtime state exactly when \(\mathsf{Ready}_V(u,Y)\) holds, and returns \(\mathsf{NoCommit}\) otherwise.

Internal identifiers may be renamed consistently, while public identifiers remain fixed.
We write \((V,q)\cong(V',q')\) when one type-preserving renaming maps the entire record and request to the other.
We use the same symbol for checker answers and runtime states when the corresponding renaming maps every field of one to the other.
An input function \(f\) is exact when indistinguishable reduced inputs always have the same correct answer up to such a renaming:
\[
\begin{aligned}
\mathsf{Exact}(f)\Longleftrightarrow {}&
\forall (V_0,q_0)\in\mathcal Q_{\rm sec}.\\[-2pt]
&\forall (V_1,q_1)\in\mathcal Q_{\rm sec}.\\[-2pt]
&(f(V_0),q_0)\cong(f(V_1),q_1)\\[-2pt]
&\Longrightarrow
\mathsf{Ans}(V_0,q_0)\cong\mathsf{Ans}(V_1,q_1).
\end{aligned}
\]

The four parts of the record are as follows.
\(\mathbf P\) contains the workflow, checkpoints, edit rules, required results, authorized removals, and the authority for each result.
\(\mathbf I\) records which workflow calls refer to the same tool action, together with each action's request, authorization label, scope, source region, current rule entry, and authorization data.
\(\mathbf E\) records processed calls, current and saved progress, prior authorizations, retries, returned values, queued requests, and the current policy state.
\(\mathbf C\) records the policy version, edit-rule version, execution record version, and active rule version, together with the internal namespace and the digest checked before new rules take effect.
Encodings, digests, authorization sequences, and returned proof data are computed from these four parts rather than treated as additional facts.

For \(J\subseteq\{\mathbf P,\mathbf I,\mathbf E,\mathbf C\}\), \(\mathsf{Keep}_J(V)\) retains the facts in \(J\).
It also removes every reference to an omitted fact and every computed value that depends on an omitted fact.
This recursive deletion has a unique output because computed values depend acyclically on the four recorded parts.
\(\mathsf{Drop}_{ca}(V)\) removes only the correspondence between workflow calls and tool actions, together with all data that could reconstruct it.
\(\mathsf{Drop}_{ra}(V)\) similarly removes only the correspondence between prior authorization records and tool actions.

\begin{lemma}[Omitting information is well defined]
\label{lem:app-omission}
Consistent renaming commutes with \(\mathsf{Keep}_J\), \(\mathsf{Drop}_{ca}\), and \(\mathsf{Drop}_{ra}\).
Their outputs therefore do not depend on the spelling of internal identifiers or the byte encoding of computed data.
\end{lemma}

\begin{proof}
A type-preserving renaming preserves references and dependencies.
It therefore maps every deletion step to the corresponding step in the renamed record.
Induction over the finite acyclic dependency order proves the claim.
\end{proof}

\begin{lemma}[Opposite answers require distinguishable records]
\label{lem:app-observational-separation}
If \(\mathsf{Ans}(V_0,q_0)\not\cong\mathsf{Ans}(V_1,q_1)\), then every exact \(f\) satisfies \((f(V_0),q_0)\not\cong(f(V_1),q_1)\).
\end{lemma}

\begin{proof}
The claim is the contrapositive of the definition of \(\mathsf{Exact}(f)\).
\end{proof}

\begin{table}[t]
\centering
\scriptsize
\setlength{\tabcolsep}{2pt}
\begin{tabular}{@{}c >{\raggedright\arraybackslash}p{0.47\columnwidth}
                  >{\raggedright\arraybackslash}p{0.34\columnwidth}@{}}
\toprule
Omitted & Record pair & Correct answers\\
\midrule
\(\mathbf P\) &
\(q_{\mathbf P}^i=\mathsf{Decide}(\mathsf{MergeSelect}(g_i,L,\sigma))\);
\(T^0=b_x{:}X\oplus_{g_0}^{\rm open}b_y{:}Y\),
\(T^1=b_x{:}X\parallel_{g_1} b_y{:}Y\). &
\(V^0\): safe;
\(V^1\): \(\mathsf{Invalid}\) because the mode is wrong.\\
\(\mathbf I\) &
\(q_{\mathbf I}=\mathsf{Decide}(\mathsf{RestoreLive}(b_L,k,\sigma_{\mathbf I}))\) after safe empty-root/ForkChoice/Checkpoint;
\(\Delta=\epsilon,\mathcal A=\Partial\{a\}\);
\(q^0_{\rm act}(x)=q^0_{\rm act}(y)=d_0\);
\(q^1_{\rm act}(x)=d_0\ne d_1=q^1_{\rm act}(y)\);
\(\ell(d_0)=\ell(d_1)=a\). &
\((X\otimes Y_k)\oplus Y\): one \Fresh, accept;
two \Fresh reject \(aa\).\\
\(\mathbf E\) &
\(u_{\mathbf E}=\mathsf{RestoreReplace}(b_L,k,\sigma_{\mathbf E})\),
\(q_{\mathbf E}=\mathsf{TryUpdate}(u_{\mathbf E},Y_0)\);
\(Y_0\) is checked at \(V^0\) after safe empty-root/ForkChoice/Checkpoint/\Fresh;
\(V^0\xrightarrow{\Dispatch(d)}V^1\), changing only whether a queued request was sent. &
\(V^0\): rules take effect;
\(V^1\): \(\mathsf{NoCommit}\) because the checked record changed.\\
\(\mathsf{Drop}_{ra}\) &
\(q_{\neg ar}=\mathsf{Decide}(\mathsf{RestoreReplace}(b_L,k,\sigma_{\neg ar}))\) after safe empty-root/ForkChoice/Checkpoint/\Fresh;
\(q_{\rm act}(x)=d\), \(\mathcal A=\Partial\{a\}\);
same calls \(\{c,x\}\), records \(\{p\}\), and actions \(\{d,e\}\);
\((q_{\rm act}(c),p)=(d,d)\) in \(V^0\), \((e,e)\) in \(V^1\). &
\(\Alias(x_k,d)\): accept;
\(\Fresh(x_k,d)\): reject \(aa\).\\
\(\mathbf C\) &
\(q_{\mathbf C}=\mathsf{TryUpdate}(u_0,Y_0)\), with \(Y_0\) computed at \(V^0\), is identical in both;
the current rule versions and dependent digests differ. &
\(V^0\): rules take effect;
\(V^1\): \(\mathsf{NoCommit}\).\\
\bottomrule
\end{tabular}
\caption{Record pairs that become indistinguishable after one fact is omitted but require opposite answers.}
\label{tab:record-pairs}
\end{table}

\subsection{A Common Initial State and Five Record Pairs}
\label{app:boot-pairs}

\begin{definition}[Checked registration]
\label{def:registration-refinement}
A typed Agent workflow \(\mathcal W\) is eligible for registration when it is a finite-state LTS whose complete paths carry result labels, end at terminal states, and yield finitely many tool-call traces after internal \(\tau\)-steps are erased.
\(\mathsf{CheckReg}\) rejects any reachable strongly connected component that can still reach a result and contains a tool-call edge within the component, permits erased \(\tau\)-cycles, and then enumerates the resulting finite traces.
\(\mathcal R\) stores a candidate map \(\lambda\), and \(\mathsf{BRuns}(\mathcal R)\) is the finite language obtained by linearizing its registered root pomsets, retaining each result label, and labeling every workflow call by its checked \Use.
\(\mathcal W\sqsubseteq_\lambda\mathcal R\) means that \(\lambda\) gives a bijection from these traces to \(\mathsf{BRuns}(\mathcal R)\) that preserves results, call order, which calls refer to which tool actions, and the signed policy authority for each source region.
For source executions \(e,e'\), write \(e\equiv_\partial e'\) when erasing internal \(\tau\)-steps yields the same tool-call trace with the same result label.
The check requires every tool action to have one certified canonical request and requires every reuse of that action to retain the same request.
It also requires \(\lambda\) to be total and one-to-one, to contain no unregistered tool calls, and to preserve prior authorizations and every required result.
\end{definition}

Every pair below starts from a checked initial state and differs in one part of the execution record.
After that part is omitted, the checker receives indistinguishable inputs even though the correct answers are opposite.

We construct the finite pairs from one empty initial state \(\mathsf S_\emptyset\).
A registered call model \(\mathcal R\) contains finite contracts, edit rules, policy, the checked registration map, and signed records for results, workflow calls, tool actions, canonical requests, source regions, scopes, and labels.
Before the root rules take effect, it contains no prior authorization, execution progress, queued request, rule version, returned value, or automaton state.
The registration relation has a step \(\mathsf S_\emptyset\xrightarrow{\mathsf{register}(\mathcal W,\mathcal R,n)}\mathsf{Reg}(\mathcal R,n)\) exactly when \(\mathsf{CheckReg}(\mathcal W,\mathcal R)\) accepts and \(n\) is a fresh private namespace.
A checked root update then initializes empty progress and request state, activates one rule version and the canonical automaton, and records signed data for every current rule entry and root result authority.

\begin{lemma}[Registration reflection]
\label{lem:registration-reflection}
For finite-state \(\mathcal W\) and finite \(\mathcal R\), \(\mathsf{CheckReg}(\mathcal W,\mathcal R)\) accepts iff the registered \(\lambda\) witnesses \(\mathcal W\sqsubseteq_\lambda\mathcal R\) in \cref{def:registration-refinement}.
On acceptance, applying \(\lambda\) to source executions modulo \(\equiv_\partial\) gives a bijection with \(\mathsf{BRuns}(\mathcal R)\).
The correspondence preserves result identity, causal order, canonical requests, prior authorizations, and the safety-check answer.
\end{lemma}
\begin{proof}
The SCC check terminates because the state graph is finite.
A reachable SCC that can still reach a result and contains a tool-call edge can be repeated before that result, producing projected words of unbounded length.
Conversely, if no such SCC contains a tool-call edge, contracting the SCCs yields a DAG with a finite, enumerable projected language.
The checker compares that language with \(\mathsf{BRuns}(\mathcal R)\).
It checks results, call order, which calls refer to which actions, signatures on those actions, that registration is total and one-to-one, that no unregistered calls appear, and prior authorizations after every partial execution.
It therefore accepts exactly when \(\mathcal W\sqsubseteq_\lambda\mathcal R\) holds.
Erasing internal steps gives a bijection between matching tool-call traces.
Canonical-request equality and deterministic \(\Resolve\) preserve prior authorizations after every partial execution.
The safe execution sets are therefore order-isomorphic, so the fixed-point operator preserves the safety decision, greatestness, and corresponding automaton.
\end{proof}

\begin{lemma}[Checked initial state]
\label{lem:app-boot}
If \(\mathsf{CheckReg}(\mathcal W,\mathcal R)\) accepts and the registered root contract and policy have a nonempty greatest safe execution set, the checked root update reaches a finite state \(S_{\mathcal R}\) with \(\mathsf{AgentSec}(S_{\mathcal R};c_{\mathcal R},\epsilon)\).
The execution record stores the checked root derivation, and every later rule update stores the checked Fork, Restore, or Merge edit or registered extension that produced it.
\end{lemma}

\begin{proof}
Registration contributes only the checked \(\mathcal R\) and \(n\), and deterministic allocation constructs every runtime identifier from them.
The root judgment and checked maps make the registered structure agree with the runtime state, and the recomputed nonempty fixed point establishes the condition on required results.
Empty runtime progress matches the executed calls, the canonical automaton starts at \(q_\epsilon\), and the atomic assignments give every current call an entry in the active rule version.
Only the initial rule update creates a root origin, and every successful later update stores the checked derivation that produced it.
\end{proof}

The following finite pairs use private names except for each displayed request and authorization label.
Each state is reached from \(\mathsf S_\emptyset\) by the stated registration, initial rule update, Checkpoint, tool-call, and edit steps.

\paragraph{\(\mathbf P\): workflow structure}
Let \(X\) and \(Y\) be singleton contracts with distinct workflow calls \(x\) and \(y\), a shared tool action, and a universal policy.
From the same registered singleton root, a checked initial rule update followed by \(\mathsf{ForkChoice}\) reaches \(V_{\mathbf P}^0\) with private group \(g_0\), while the corresponding \(\mathsf{ForkParallel}\) trace reaches \(V_{\mathbf P}^1\) with private group \(g_1\).
The two records differ only in
\[
 V_{\mathbf P}^0.T=b_L{:}X\oplus_{g_0}^{\mathsf{open}}b_R{:}Y,
 \qquad
 V_{\mathbf P}^1.T=b_L{:}X\parallel_{g_1} b_R{:}Y .
\]
Use the same registered \(\mathsf{MergeSelect}\) schema with \(q_{\mathbf P}^i=\mathsf{Decide}(\mathsf{MergeSelect}(g_i,L,\sigma))\).
The choice record has a unique edit derivation and a nonempty one-step safe execution set.
The parallel record has no \(\mathsf{MergeSelect}\) derivation because its group mode is wrong.
Omitting \(\mathbf P\) removes the constructor tag and all computed signed data that depends on it.
The private renaming \(g_0\mapsto g_1\) then makes the remaining record-request pairs indistinguishable.

\paragraph{Common start for the correspondence pairs}
Let \(\mathbf 1\) be the singleton-result contract whose unique pomset is empty.
Its completion \(\epsilon\) gives a nonempty greatest safe execution set under \(\mathcal A=\Partial\{a\}\).
Both pairs below put these safe root rules into effect and then apply a checked \(\mathsf{ForkChoice}\).
Any inactive edit-rule entry needed to equalize the individual tool actions is identical in both records and never becomes active.

\paragraph{\(\mathbf I\): which calls refer to the same action}
Let \(X,Y\) be singleton contracts with workflow calls \(x,y\), and let \(d_0,d_1\) be tool actions without authorization records, both labeled \(a\).
Two checked call models contain the same workflow and the same individual calls and actions, but the calls refer to the actions differently:
\[
\begin{array}{c|cc}
 &q_{\rm act}(x)&q_{\rm act}(y)\\ \hline
 V_{\mathbf I}^0&d_0&d_0\\
 V_{\mathbf I}^1&d_0&d_1 .
\end{array}
\]
From the \(\mathbf 1\) root now in force, the common Fork rule derives \(b_L{:}X\oplus_g^{\mathsf{open}}b_R{:}Y\).
Each result uses one tool action, so both Fork updates are safe under \(\Partial\{a\}\).
Create a checkpoint for \(b_R\) before progress and issue the common query \(q_{\mathbf I}=\mathsf{Decide}(\mathsf{RestoreLive}(b_L,k,\sigma_{\mathbf I}))\).
Up to clone renaming, the target is \((X\otimes Y_k)\oplus Y\).
In \(V_{\mathbf I}^0\), the left result uses one \Fresh and one \Alias on \(d_0\), while the bypass result uses one \Fresh, so every result has authorization sequence \(a\).
In \(V_{\mathbf I}^1\), every left-result linearization uses \Fresh on both \(d_0,d_1\), so its word is \(aa\notin\mathcal A\).
The left result is compatible with \(\epsilon\), so the first \(\widehat\Phi\) round also removes the otherwise safe bypass completion and leaves the greatest fixed point empty.
Applying \(\mathsf{Drop}_{ca}\) makes the remaining records indistinguishable because their individual calls and actions were already identical.

\paragraph{\(\mathbf E\): authorization and execution progress}
Let \(C,X\) be singleton contracts with workflow calls \(c,x\), let \(q_{\rm act}(c)=q_{\rm act}(x)=d\), and label \(d\) by \(a\).
From the safe \(\mathbf 1\) root, a checked Fork update puts \(b_L{:}C\oplus_g^{\mathsf{open}}b_R{:}X\) into effect.
Create a checkpoint for \(b_R\), then use \(c\) to complete the left arm, create authorization record \(p\), and queue its canonical request.
At the resulting \(V_{\mathbf E}^0\), check \(u_{\mathbf E}=\mathsf{RestoreReplace}(b_L,k,\sigma_{\mathbf E})\) to obtain \(Y_0\).
Its clone of \(x\) reuses \(d\), so the update is safe under \(\mathcal A=\Partial\{a\}\).
Let \(V_{\mathbf E}^0\xrightarrow{\Dispatch(d)}V_{\mathbf E}^1\), which changes only whether the queued request has been sent and leaves \(\mathbf P,\mathbf I,\mathbf C\) unchanged.
For the same request \(q_{\mathbf E}=\mathsf{TryUpdate}(u_{\mathbf E},Y_0)\), the rules take effect at \(V_{\mathbf E}^0\), while \(V_{\mathbf E}^1\) returns \(\mathsf{NoCommit}\) because the digest of the checked record has changed.

\paragraph{\(\mathsf{Drop}_{ra}\): which action a prior authorization names}
For a separate pair, keep the common active row \(q_{\rm act}(x)=d\), creator \(c\), authorization record \(p\), and tool actions \(d,e\) labeled \(a\), but set \(q_{\rm act}(c)=d\) and \(p\mapsto d\) in \(V_{\neg ar}^0\), versus \(q_{\rm act}(c)=e\) and \(p\mapsto e\) in \(V_{\neg ar}^1\).
Both states follow the same safe empty-root/ForkChoice/Checkpoint/\Fresh shape above, have authorization sequence \(a\), and use the common query \(q_{\neg ar}=\mathsf{Decide}(\mathsf{RestoreReplace}(b_L,k,\sigma_{\neg ar}))\).
The cloned \(x_k\) is \(\Alias(x_k,d)\) in \(V_{\neg ar}^0\), but \(\Fresh(x_k,d)\) produces forbidden word \(aa\) in \(V_{\neg ar}^1\).
The two records contain the same creator, authorization record, workflow call, and tool actions, and both retain \(x\mapsto d\).
\(\mathsf{Drop}_{ra}\) removes the only differing correspondence and every value that could reconstruct it, so the remaining records are indistinguishable.
This pair proves that an exact checker must know which tool action each prior authorization names.

\paragraph{\(\mathbf C\): the active rule version}
Let an initial namespace be a family \(n=(n_s)_{s\in\mathsf{Sort}}\) of injective, per-sort fresh supplies.
From the same empty initial state, register the same public call model \(\mathcal R\) with supplies \(n^0,n^1\) that agree on every type except the rule-version type.
The checked root updates create distinct active versions \(\eta_0\ne\eta_1\).
The resulting \(V_{\mathbf C}^0,V_{\mathbf C}^1\) have identical \(\mathbf P,\mathbf I,\mathbf E\), while their active rule versions and dependent digests differ.
Check the same registered edit \(u_0\) at \(V_{\mathbf C}^0\) to obtain \(Y_0=\mathsf{Accepted}(\delta_0,W_0,\widehat B_0,Z_0)\), then submit \(q_{\mathbf C}=\mathsf{TryUpdate}(u_0,Y_0)\) to both records.
At \(V_{\mathbf C}^0\), the final check reproduces \(\mathsf{cut}_{Z_0}\), so the rules take effect.
At \(V_{\mathbf C}^1\), the active version is \(\eta_1\) rather than the \(\eta_0\) recorded in \(Z_0\), so the update returns \(\mathsf{NoCommit}\).
Omitting \(\mathbf C\) also removes dependent rule entries, authorization data, and digests, leaving the record-request pairs indistinguishable.

\begin{proposition}[Four answer-relevant record parts]
\label{prop:app-explicit-pairs}
Every pair above consists of finite runtime states reached from the common empty initial state and satisfying \(\mathsf{AgentSec}\).
After the corresponding information is omitted, each pair becomes indistinguishable even though the correct answers are opposite.
Therefore every \(\mathsf{Keep}_J\) with \(J\subsetneq\{\mathbf P,\mathbf I,\mathbf E,\mathbf C\}\), as well as \(\mathsf{Drop}_{ca}\) and \(\mathsf{Drop}_{ra}\), is inexact.
Every exact checker input must distinguish all five pairs.
\end{proposition}

\begin{proof}
\Cref{lem:app-boot} gives the common initial state, and the traces above use only modeled Checkpoint, tool-call, Fork/Restore/Merge, and rule-update steps.
For the two correspondence pairs, the empty root completes with \(\epsilon\), and each Fork arm has a one-\Fresh completion labeled \(a\).
Hence every displayed state before its request is reachable under \(\Partial\{a\}\).
Open choice enables the Checkpoint steps, and the active automata enable each stated left-arm \Fresh before that arm becomes a completed, still-addressable branch.
For the \(\mathbf E\) pair, \(\Dispatch\) preserves \(\mathbf P,\mathbf I,\mathbf C\) but changes the checked record, giving a successful update versus \(\mathsf{NoCommit}\).
The displayed edit-rule and resolution calculations give the other opposite answers.
\Cref{lem:app-omission} and the stated private renamings make the reduced inputs indistinguishable.
Each omission collapses its corresponding pair, so \cref{lem:app-observational-separation} proves inexactness.
\end{proof}

Giving the \(\mathbf P\) pair the same caller description of what should happen next and the same completion conditions leaves their opposite answers unchanged.
The target alone is therefore insufficient.
A generalized nonblocking solver becomes applicable only after the trusted runtime has derived the edited workflow, preserved result identities, and the completion condition for each result.

\subsection{Ideal Atomic Execution Machine}
\label{sec:ideal}

An ideal state \(I=(\kappa,H,\Delta,\mathsf{out},\mathcal A,c,r)\), with \(c=(H_c,u,L^*,\widehat B^*)\), contains the execution record, current rule version and entries, checked safe execution set, and current partial execution.
Write \(\Theta_I(u)=(\kappa,H,\Delta,\mathsf{out},\mathcal A,u)\).
The runtime state additionally represents the automaton, transaction phase, authorization data, and a prepared rule update.

\paragraph{Edit transition}
For current \(\Theta_I(u)\), an undefined derivation gives state-preserving \(\mathsf{Invalid}\), while a defined derivation with no \(\mathsf{MaxSafe}\) witness gives state-preserving \(\mathsf{Reject}\).
Otherwise \(\mathsf{Derive}(\Theta_I(u))\) supplies the target and rule version, and independently defined \(\mathsf{SafeBehavior}(\Theta_I(u))=(L^*,\widehat B^*)\) enables one atomic \(\mathsf{Edge}(u,[\mathcal H_u^+]_{\cong})\) that makes them current and resets \(r\) to \(\epsilon\).
All three answers are computed from the same current execution record and are mutually exclusive and exhaustive.

\paragraph{Tool call}
For submitted tool request \(m\), an enabled \(x\) passes the ideal automaton exactly when its current rule entry and authorization data identify tool action \(d\), \(\mathsf{canon}_\kappa(m)=\mathsf{inv}_\rho(d)\), and \(ra\in L^*\).
Here \(a=\Fresh(x,d)\) if \(d\) has no authorization record, and \(a=\Alias(x,d)\) otherwise.
One atomic transition applies \(\mathsf{HStep}\), appends \(a\) to \(\chi\), advances \(\mathsf{ver}\), removes \(x\)'s current rule entry, and sets \(r:=ra\).
\Fresh also creates the labeled authorization record and queues the canonical request, while \Alias links to the existing return record.
Any failed premise returns \(\mathsf{deny}(x)\) without changing state or inventing a return value.

\paragraph{Checkpoint, retry, sending, and restart}
\(\mathsf{Checkpoint}(k,b)\), verified retry, \(\Dispatch\), completion updates, and restart after a stop have the state changes in \cref{tab:app-event-matrix}.
The ideal and runtime machines apply the same freshness, current-branch, and remaining-work checks, so they make the same Checkpoint decision.
This transition system represents external-network completion after \(\Dispatch\) as environmental behavior.

By \cref{lem:largest-safe-family}, every reachable partial execution of the ideal machine retains a policy-safe completion until the next edit transition changes the workflow contract.
The runtime implements this machine without enumerating the full completion set at every call.

\subsection{Runtime State and Atomic Transitions}
\label{app:runtime-transitions}

For \(\Theta=(\kappa,H,\Delta,\mathsf{out},\mathcal A,u)\), let \(C_{\rm der}(\Theta)\) be the canonical record of the failed derivation.
\[
\begin{aligned}
\mathsf{Check}(\Theta)
 &=\mathsf{Invalid}([\gamma(\Theta),C_{\rm der}(\Theta)]_{\cong}),\\[-2pt]
 &\qquad\text{if }\mathsf{Derive}(\Theta)\text{ is undefined};\\
\mathsf{Check}(\Theta)
 &=\mathsf{Reject}([\gamma(\Theta),C_{\rm fp}]_{\cong}),\\[-2pt]
 &\qquad\text{if }\mathsf{Derive}(\Theta)\text{ is defined and }\
   \widehat B^\dagger_{H,u}=\varnothing.
\end{aligned}
\]
When \(\mathsf{Derive}(\Theta)\) is defined and
\(\widehat B^\dagger_{H,u}\ne\varnothing\),
\[
\mathsf{Check}(\Theta)=
\mathsf{Accepted}(\delta,W^\dagger_{H,u},
 \widehat B^\dagger_{H,u},Z).
\]
where
\[
 Z=(\omega,\zeta,\kappa_Z,\mathsf{cut}_Z,u,
    H_{\mathsf{aut}},\eta^-,\eta^+).
\]
\(\mathsf{Invalid}\) carries the failed derivation, \(\mathsf{Reject}\) carries the complete ordered rejection proof, and \(\mathsf{Accepted}\) carries the unique derivation, coverage proof, canonical automaton, and \(Z\).
\(Z\) binds the checked execution record, derived target, automaton, and both old and new rule versions.

The runtime state for one policy domain is
\[
S=(\kappa,H,\Delta,\mathsf{out},\mathcal A,
    \mathsf{phase},\mathsf{entry},\mathsf{aut},\mathsf q,\mathsf{checked}),
\]
where \(\mathsf{phase}\) records whether each rule version \(\eta\) is unallocated (\(\bot\)), inactive, active, or closed.
The initial rule runs only after \(\mathsf{CheckReg}(\mathcal W,\mathcal R)\) accepts, and its exact construction appears in \cref{app:boot-pairs}.
Each active version stores \(c=(\mathsf{record}_c,H_c^+,u,W,\widehat B,Z)\), which contains the checked execution record, derived target, edit, safe partial executions, complete executions labeled by result, and data needed for the final recheck.

\begin{definition}[Agent runtime invariant]
\label{def:agent-sec}
For \(c=(\mathsf{record}_c,H_c^+,u,W,\widehat B,Z)\), \(\mathsf{record}_c=(\kappa_c,H_c,\Delta_c,\mathsf{out}_c,\mathcal A_c)\), and resolved \(r\), \(\mathsf{AgentSec}(S;c,r)\) holds exactly when the following clauses hold.
\begin{enumerate}[leftmargin=*,itemsep=0pt,topsep=1pt]
\item \emph{Checked record and edit rule.}
The state is well formed, \(c=\mathsf{checked}(H.\eta)\), its origin derives the target now in force, and every required result has one signed authority record and a checked \(\mathsf{Covers}_u\) proof.
\item \emph{Required results.}
\(\widehat B=\widehat B^\dagger_{H_c,u}\ne\varnothing\), \(W=\Partial(\pi_2\widehat B)\), and \(Z\) verifies both against \(\mathsf{record}_c\).
\item \emph{Executed calls.}
\((H.G,H.T)=\mathsf{HRes}((H_c^+.G,H_c^+.T),r)\), \(H.\chi=H_c^+.\chi r\), and authorization records, queue entries, return links, and checkpoints match \(r\) and advance only monotonically.
\item \emph{Current automaton.}
\(\mathsf{aut}(H.\eta)\cong\mathsf{Can}(W,\pi_2\widehat B)\), and \(\mathsf q(H.\eta)\) is its derivative after \(r\).
\item \emph{Current rule version.}
Only \(H.\eta\) is active, \(\mathsf{entry}|_{X_{\sem H.T}}=H.\beta\), and it provides exactly one signed rule entry for each current workflow call, while inactive or closed versions provide none.
\end{enumerate}
\end{definition}

\paragraph{Atomic tool call}
For submitted tool request \(m\), set \(m^\circ=\mathsf{canon}_\kappa(m)\).
For a call \(x\) with current rule entry \(j\) in version \(\eta\), \(\mathsf{VerifyToken}(h,\rho(x),j,d,m^\circ)\) verifies the issuing runtime, scope, rule entry, tool action, source region, and signed digest and requires \(m^\circ=\mathsf{inv}_\rho(d)\).
Set \(a=\Fresh(x,d)\) if \(d\notin D_\Delta\), otherwise \(a=\Alias(x,d)\), and let \(\Use_x(m)=\mathsf{Use}(x,j,h,m)\).
Its sole state-changing transition is
\begin{equation}
\label{eq:use-rule}
\frac{
\begin{gathered}
 x\text{ enabled in }T\quad
 \mathsf{entry}(x)=(j,\eta)\\[-2pt]
 \eta=H.\eta\quad
 \mathsf{phase}(\eta)=\mathsf{active}\\[-2pt]
 \mathsf{VerifyToken}(h,\rho(x),j,d,m^\circ)\quad
 \mathsf q(\eta)\xrightarrow{a}_{\mathsf{aut}(\eta)}q'
\end{gathered}}
{S\xrightarrow{\Use_x(m)/a}S'} .
\end{equation}
The paired update \(\mathsf{HStep}((G,T),a)\) removes \(x\), records any choice, preserves the other workflow constructors, and appends \(a\) to its branch progress.
\(\mathsf{Advance}(S,x,a,q')\) changes exactly
\[
\begin{aligned}
 (H.G,H.T)&:=\mathsf{HStep}((G,T),a),&
 H.\chi&:=\chi a,\\
 H.\mathsf{ver}&:=\mathsf{ver}+1,& H.\beta&:=\beta\setminus\{x\},\\
 \mathsf{entry}&:=\mathsf{entry}\setminus\{x\},&
 \mathsf q(\eta)&:=q'.
\end{aligned}
\]
In addition, \Fresh appends the immutable labeled authorization record and queues the canonical request \((d,\mathsf{inv}_\rho(d))\), never an unchecked caller request.
\Alias records the tool action's stored return link without changing \(\Delta\), and all other fields remain unchanged.
A verified \Retry follows the stored link from \(\chi\) to the authorization record for the same canonical request and returns any stable value without changing state.

\paragraph{Final check and atomic rule update}
At the atomic update point, the trusted runtime ignores any workflow or rule version supplied by the caller, sets \(\Theta_S(u)=(\kappa,H,\Delta,\mathsf{out},\mathcal A,u)\), and re-runs \(\mathsf{Derive}(\Theta_S(u))\) from the current execution record.
For a defined derivation, \(\mathsf{ReCut}(S,u,Z)\) recomputes the right-hand side of \cref{eq:checked-record-digest}, while \(\mathsf{Verify}_Z\) recomputes the fixed point, coverage proof, canonical automaton, and both digests.
The rules may take effect only if
\begin{equation}
\label{eq:update-premise}
\begin{gathered}
 \kappa=\kappa_Z,\qquad
 \mathsf{ReCut}(S,u,Z)=\mathsf{cut}_Z,\\
 \mathsf{Verify}_Z(W^\dagger_{H,u},B^\dagger_{H,u},
    \widehat B^\dagger_{H,u}),\\
 H.\eta=\eta^-,\qquad
 \mathsf{phase}(\eta^+)\in\{\bot,\mathsf{inactive}\}.
\end{gathered}
\end{equation}
\(\mathsf{Ready}(S,u,Y)\) holds when \(Y=\mathsf{Accepted}(\delta,W^\dagger,\widehat B^\dagger,Z)\), the current derivation is \(\delta\), and all state-dependent premises of \cref{eq:update-premise} hold.
The committing domain transaction makes \((\kappa_u,H_u,\mathcal A_u)\) current, closes \(\eta^-\), activates \(\eta^+\), assigns a rule entry to every target workflow call, starts the canonical automaton at \(q^0\), and stores the checked answer.
For this successor, write
\[
c_\delta^+=
(\mathsf{state}(\Theta_S(u)),H_u,u,
 W^\dagger_{H,u},\widehat B^\dagger_{H,u},Z).
\]
Only after activation does it expose \(\mathcal H_u^+\) and emit \(\mathsf{Edge}(u,[\mathcal H_u^+]_{\cong})\).
Old authorization data then fails, while \Retry still follows its stored authorization-record link.
A stale answer returns \(\mathsf{NoCommit}\) and must be recomputed from the current execution record.

\subsection{Atomic Events Preserve the Invariant}
\label{app:event-preservation}

Write \(A_1,\ldots,A_5\) for the five clauses of \(\mathsf{AgentSec}\) in their stated order.
The following matrix covers every state-changing transition after the initial rules take effect.
Initial registration occurs before the root safety check, and later additions use the registered-extension transition below.

\begin{table*}[t]
\centering
\scriptsize
\setlength{\tabcolsep}{3pt}
\begin{tabular}{@{}p{0.15\textwidth}p{0.29\textwidth}p{0.12\textwidth}p{0.36\textwidth}@{}}
\toprule
Event & Atomic change & New state & Why the invariant remains true\\
\midrule
\Fresh use &
Apply \(\mathsf{HStep}\) to \((G,T)\); append \(\Fresh(x,d)\) to global \(\chi\); advance \(\mathsf{ver},q\); append one authorization record and queued canonical request; remove the current rule entry. &
\((c,r\Fresh(x,d))\) &
\(A_1\): record and edit rules unchanged.
\(A_2\): rechecking after the partial execution.
\(A_3\): unique authorization record and request order.
\(A_4\): canonical derivative.
\(A_5\): remove exactly \(x\)'s rule entry.\\
\Alias use &
Apply \(\mathsf{HStep}\) to \((G,T)\); append \(\Alias(x,d)\) to global \(\chi\); advance \(\mathsf{ver},q\); add only the return-record link; remove the current rule entry. &
\((c,r\Alias(x,d))\) &
\(A_1\): the signed record still names the same tool action.
\(A_2\): rechecking after the partial execution.
\(A_3\): link names an earlier authorization record and the authorization sequence is unchanged.
\(A_4\): canonical derivative.
\(A_5\): remove exactly \(x\)'s rule entry.\\
Successful edit transition, six edit forms or registered extension &
Make derived \((\kappa^+,H^+,\mathcal A^+)\) and the checked automaton current; close \(\eta^-\); activate \(\eta^+\); publish new rule entries. &
\((c_\delta^+,\epsilon)\) &
\(A_1\): edit-rule correctness and result preservation.
\(A_2\): exact checker.
\(A_3\): prior authorizations and queued requests retained.
\(A_4\): start at \(q_\epsilon\).
\(A_5\): complete rule update.\\
Verified \(\mathsf{Invalid}\) or \(\mathsf{Reject}\) &
Return the failed derivation or rejection proof; no runtime-state change. &
\((c,r)\) &
All clauses unchanged; \(\gamma\) ties the answer to the checked execution record.\\
\(\mathsf{Checkpoint}(k,b)\) &
Record the current workflow and progress in a signed checkpoint; advance \(\mathsf{ver}\). &
\((c,r)\) &
\(A_1\): immutable well-typed extension.
\(A_2,A_4,A_5\): unchanged.
\(A_3\): exactly the permitted \(K\) extension.\\
\(\mathsf{Preload}\) &
Prepare content only in an inactive rule version assigned to no call. &
\((c,r)\) &
\(A_1\)--\(A_4\): unchanged.
\(A_5\): the version remains inactive and unavailable to calls.\\
Retry, retrieval, delivery &
Emit \(\mathsf{Return}(t,[v]_{\cong})\) for a stored value whose signature verifies; no runtime-state change. &
\((c,r)\) &
All clauses unchanged; retry additionally checks the same canonical request and prior-authorization link.\\
Denial &
Return \(\mathsf{deny}(x)\); no state change. &
\((c,r)\) &
All clauses unchanged.\\
\(\Dispatch(d)\) &
Mark exactly the oldest queued canonical request as sent. &
\((c,r)\) &
\(A_3\): sent requests are exactly the earliest queued requests, in order.
Other clauses unchanged.\\
Completion update &
Emit \(\mathsf{Settle}(d,[v]_{\cong})\), advance one authorization record phase, and fill its stable return record. &
\((c,r)\) &
\(A_3\): tool action, canonical request, creator, authorization label, and order are immutable.
Other clauses unchanged.\\
Stale or malformed answer &
Return \(\mathsf{NoCommit}\); no runtime-state change. &
\((c,r)\) &
All clauses unchanged.\\
Stop and restart &
Discard unfinished in-memory work and expose the last committed runtime state. &
\((c,r)\) &
All clauses unchanged because each atomic update appears either completely or not at all.\\
\bottomrule
\end{tabular}
\caption{Exhaustive event classes and preservation of the five \(\mathsf{AgentSec}\) clauses.}
\label{tab:app-event-matrix}
\end{table*}

\begin{lemma}[Atomic events are well defined]
\label{lem:app-event-derivative}
The atomic update induces a partial function \(\partial\) on runtime-state and event pairs modulo consistent renaming, written \(\langle V,e\rangle_{\cong}\).
If \(\mathsf{AgentSec}(S;c,r)\) and \(S\xrightarrow eS'\), the witness in \cref{tab:app-event-matrix} satisfies \(\mathsf{AgentSec}(S';c',r')\) and
\[
 [\mathsf{Sec}(S';c',r')]_{\cong}
 =\partial(\langle\mathsf{Sec}(S;c,r),e\rangle_{\cong}).
\]
\end{lemma}

\begin{proof}
All transition rules use typed equality, membership, order, canonical constructors, signature verification, and equality checks.
These operations are unchanged by consistently renaming private identifiers in the state and event.
For an event \(e\), let \(G_e(V)\) be its precondition and \(\mathsf U_e(V)\) its successor update when that precondition holds.
For every such renaming \(\pi\), \(G_e(V)\Longleftrightarrow G_{\pi e}(\pi V)\) and \(\mathsf U_{\pi e}(\pi V)=\pi\mathsf U_e(V)\).
Thus consistently renamed inputs have the same definedness and consistently renamed successors.
This proves that \(\partial\) is well defined.

For \Fresh and \Alias, \cref{lem:resolution,lem:recheck-coherence,lem:runtime-language} establish the first two rows of the matrix.
If \(x\) lies at \(b\), \(\mathsf{HStep}\) appends \(a\) to \(b\)'s progress, while the same atomic tool-call update appends \(a\) to the global resolved-call trace \(\chi\).
Associativity of language quotients gives
\[
 (\Gamma_b/\mathsf{raw}(\chi_b))/x
 =\Gamma_b/\mathsf{raw}(\chi_ba).
\]
Thus the updated leaf and any later Checkpoint remain synchronized with \(G\).
For a successful edit transition, \cref{lem:history-derivation,lem:live-extension,lem:app-spec-bridge,lem:runtime-language,lem:cut-linearization} establish the third row for all six edit forms and registered extension.
Each remaining rule changes only the fields shown in the matrix, and its row checks every invariant clause whose fields change.
Thus every successor has the displayed new state and preserves \(\mathsf{AgentSec}\).
\end{proof}

\begin{corollary}[Finite-sequence preservation]
\label{cor:app-finite-sequence}
Starting from any valid bootstrap, every finite sequence of the event classes in \cref{tab:app-event-matrix} ends in a state satisfying \(\mathsf{AgentSec}\).
The sequence length is unbounded, although every individual safety-check instance and event payload is finite.
For a finite word \(\bar e\), the iterated update is defined on \(\langle V,\bar e\rangle_{\cong}\), where one renaming acts on the initial state and the whole word.
\end{corollary}

\begin{proof}
The base case is \cref{lem:app-boot}, and \cref{lem:app-event-derivative} supplies the invariant-preserving inductive step and consistency under renaming the remaining events.
\end{proof}

\subsection{Events Outside the Agent's Control and Restarts}
\label{app:restart}

The resolved alphabet contains only tool calls checked by the runtime automaton.
Agent request arrival, scheduling, completion, delivery, and stop or restart events do not append a \Fresh or \Alias symbol.
Request arrival and scheduler choice do not change committed runtime state.
Completion updates and delivery emit modeled API observations but leave the checked partial execution unchanged.
\(\Dispatch\) is observable only when the oldest queued request is recorded as sent, and it also leaves the checked partial execution unchanged.
Therefore none of these events advances the automaton for tool calls.

\begin{lemma}[Other events preserve the checked partial execution]
\label{lem:app-uncontrollable}
If \(\mathsf{AgentSec}(S;c,r)\), \(e\) is one of the modeled events above, and \(S\xrightarrow eS'\), then the checked partial execution and canonical automaton state are unchanged and \(\mathsf{AgentSec}(S';c,r)\) holds.
\end{lemma}

\begin{proof}
Request arrival and scheduling do not change committed state, while denial, retry, retrieval, delivery, and restart preserve the checked runtime state.
Request sending and completion updates change only the queued-request status, authorization-record status, or stored return record permitted by the executed-calls clause of \(\mathsf{AgentSec}\).
None changes \(H.T\), \(H.\chi\), the authorization sequence derived from \(\Delta\), \(r\), or \(\mathsf q(\eta)\), and the corresponding row of \cref{tab:app-event-matrix} preserves every other invariant clause.
\end{proof}

The restart lemma uses the linearizable atomic transactions required by \cref{sec:model}.
Let \(D\) be the committed runtime state, \(v\) the unfinished in-memory state, and \(\bar t=t_1\cdots t_n\) the finite sequence of domain transactions overlapping a stop, ordered by their linearization points.
Atomicity selects \(\bar t_{\leq k}=t_1\cdots t_k\), containing exactly the transactions that committed before the stop, and gives
\[
\mathsf{Recover}(D,v,\bar t)
 =\mathsf{commit}_{t_k}\circ\cdots\circ
   \mathsf{commit}_{t_1}(D),
\]
with \(k=0\) interpreted as \(D\).
No state reached after restart contains only part of one modeled atomic update.
In particular, a \Fresh transaction cannot expose an authorization record without its global trace, branch progress, and queued request.
A rule update cannot activate a new version without closing the old version and assigning every target rule entry.

\begin{lemma}[Restart preserves atomic commits]
\label{lem:app-restart}
Let \(S_0\to S_1\to\cdots\to S_n\) be the modeled path whose atomic transitions correspond to \(t_1,\ldots,t_n\) in the order above.
If committed state \(D\) refines \(S_0\), then restarting after exactly \(t_1,\ldots,t_k\) have committed yields a state refining \(S_k\).
At committed-state boundaries, stopping and restarting without another commit is consequently a \(\tau\) self-loop.
\end{lemma}

\begin{proof}
Induction on \(k\) starts from \(D\) refining \(S_0\).
At each step, the complete atomic update \(\mathsf{commit}_{t_i}\) is the update in the corresponding row of \cref{tab:app-event-matrix}, so it maps a concrete state refining \(S_{i-1}\) to one refining \(S_i\).
The recovery equation applies exactly the first \(k\) updates and no partial update, yielding the claim.
Once the LTS records the restarted state, stopping again changes only unfinished memory and is therefore a self-loop.
\end{proof}

\subsection{Agent-API Weak Bisimulation After Restart}
\label{app:bisimulation}

We now prove \cref{thm:runtime-correctness} by exhausting the observable and silent cases.
For runtime state \(S=(\kappa,H,\Delta,\mathsf{out},\mathcal A,\ldots)\), define \(\mathsf{Core}(S)=(\kappa,H,\Delta,\mathsf{out},\mathcal A)\) and \(\mathsf{Current}(H)=X_{\sem{H.T}}\).
For \(c=(\mathsf{record}_c,H_c^+,u,W,\widehat B,Z)\), define \(\mathsf{icut}(c)=(H_c^+,u,W,\widehat B)\).
For \(I=(K_I,c_I,r_I)\) with \(K_I=(\kappa_I,H_I,\Delta_I,\mathsf{out}_I,\mathcal A_I)\), let
\[
\begin{aligned}
\alpha_I(I)&=(K_I,c_I,r_I,H_I.\eta,H_I.\beta),\\
\alpha_S(S;c,r)&=(\mathsf{Core}(S),\mathsf{icut}(c),r,
 H.\eta,\mathsf{entry}|_{\mathsf{Current}(H)}).
\end{aligned}
\]
Define \(I\,\mathcal B\,S\) iff some \(c,r\) satisfy \(\mathsf{AgentSec}(S;c,r)\) and \(\alpha_I(I)=\alpha_S(S;c,r)\).
Write \(\xRightarrow{\ell}=\tau^*\ell\tau^*\) and \(\xRightarrow{\tau}=\tau^*\).
Fix \(I\,\mathcal B\,S\) with witness \((c,r)\).
Equality of \(\alpha_I(I)\) and \(\alpha_S(S;c,r)\) gives both machines the same \(\kappa\), current policy, execution record, authorization log, request queue, stored return records, checked safe execution set, current partial execution, rule version, and current rule entries.

\paragraph{Edit transitions}
The equality makes both machines form the same current \(\Theta_I(u)\) and record digest \(\gamma(\Theta_I(u))\).
If derivation is undefined, both sides return the same canonical \(\mathsf{Invalid}\) answer and leave their states unchanged.
If derivation exists but no safe execution set exists, \cref{lem:app-spec-bridge} and deterministic fixed-point checking make both sides return the same canonical \(\mathsf{Reject}\) answer and leave their states unchanged.
Otherwise the ideal machine enables its edit transition and, by \cref{lem:app-spec-bridge}, the runtime checker accepts the same edit and computes the same largest safe execution set.
The runtime takes an explicit \(\mathsf{Preload}\) transition that records the prepared answer in an inactive rule version, exposes no authorization data or Agent-visible answer, and preserves \(\mathsf{AgentSec}\).
The ideal side takes its single atomic edge, and both successors have the same \((\kappa_u,H_u,\mathcal A_u)\), checked safe execution set, empty processed-call sequence, active rule version, and target rule entries.
Deterministic allocation gives consistently renamed authorization data, so both expose \(\mathsf{Edge}(u,[\mathcal H_u^+]_{\cong})\).
Conversely, every successful runtime rule update has passed the edit-rule, fixed-point, and record-digest checks, so \cref{lem:app-spec-bridge} enables the matching ideal edge.
This argument applies uniformly to all six forms of Fork, Restore, and Merge and to the registered extension.

\paragraph{Tool calls and submitted requests}
For a submitted tool request \(m\), both automata compute the same \(m^\circ=\mathsf{canon}_\kappa(m)\).
The runtime authorization check verifies the issuing runtime, scope, workflow call, rule entry, tool action, source region, signed digest, current rule version, and equality \(m^\circ=\mathsf{inv}_\rho(d)\).
The ideal automaton requires the same rule entry, tool action, and request equality.
Both sides then inspect the same set \(D_\Delta\) of tool actions with authorization records, so they choose the same \(a\in\{\Fresh(x,d),\Alias(x,d)\}\).
By \cref{cor:app-exact-boundary,lem:runtime-language}, the runtime automaton enables \(a\) after \(r\) exactly when \(ra\in L^*\) in the ideal machine.

In the \Fresh case, both successors apply the same workflow update, append the same symbol to the global trace and branch progress, advance \(\mathsf{ver}\), create the same authorization record, and queue the same canonical request.
The runtime request queue stores \((d,\mathsf{inv}_\rho(d))\), while the ideal queue stores \((d,m^\circ)\).
The verified request equality makes these entries identical.
In the \Alias case, both successors apply the same workflow update, append the same symbol to the global trace and branch progress, and link the workflow call to the same stored return record without changing \(\Delta\) or the request queue.
Both successors remove \(x\)'s current rule entry.
Because the remaining workflow no longer contains \(x\), their entries for current calls remain equal.
In either case, \cref{lem:app-event-derivative} re-establishes \(\mathcal B\) with witness \((c,ra)\).

If any workflow call, rule entry, authorization datum, canonical request, rule version, branch, or automaton transition check fails, both automata fail.
Both machines emit \(\mathsf{deny}(x)\), preserve state, and remain related.

\paragraph{Races between a tool call and a rule update}
Tool calls and rule updates serialize on the same execution record and active rule version.
If a \Fresh use commits first, it changes \((G,T,\Delta,\chi,\mathsf{ver},\mathsf{out})\), all of which are covered by the digest of the checked record, so the prepared rule update becomes stale.
If an \Alias use commits first, it changes \((G,T,\chi,\mathsf{ver})\) even though \(\Delta\) and the request queue are unchanged, so the prepared rule update again becomes stale.
If the rule update commits first, it atomically closes \(\eta^-\), activates \(\eta^+\), refreshes every current rule entry and its authorization data, and exposes the new authorization data only after activation.
The old \Fresh or \Alias use then fails the current-rule-version automaton and is denied without progress.
The ideal atomic order has exactly the same two cases, so neither side allows a third race outcome.

\paragraph{Replies, sent requests, and silent steps}
A retry with valid authorization data, retrieval of that data, or delivery of a returned value leaves the runtime state unchanged.
Both machines follow the same stored link and expose \(\mathsf{Return}(t,[v]_{\cong})\).
\(\Dispatch(d)\) is enabled for the same oldest queued request and makes the same observable request-sent update.
\(\mathsf{Settle}(d,[v]_{\cong})\) is enabled for the same authorization record and produces the same stable returned value and normalized update.
\(\mathsf{Checkpoint}(k,b)\) takes matching visible steps and yields the same decision.
Runtime preload and internal \(\mathsf{NoCommit}\) answers are \(\tau\)-steps matched by zero ideal steps because \(\alpha_S\) omits inactive prepared data and \(\mathsf{obs}_{\rm api}\) hides internal runtime answers.
A stop and restart is matched by a \(\tau\) self-loop on each committed-state machine by \cref{lem:app-restart}.

\begin{theorem}[Agent-API weak bisimulation after restart, detailed]
\label{thm:app-bisimulation}
\(\mathcal B\) is a divergence-insensitive weak bisimulation under \(\mathsf{obs}_{\rm api}\).
\end{theorem}

\begin{proof}
The transition rules in \cref{sec:ideal,sec:enforcement} and the matrix in \cref{tab:app-event-matrix} cover every possible case.
For every step from either side, the corresponding case constructs a path with observation \(\mathsf{obs}_{\rm api}(\ell)\) and a successor related by the witness in that matrix.
The construction is symmetric because the ideal and runtime automata accept the same tool calls, while runtime-only steps preserve their common visible state.
Arbitrary repetition of silent preload, internal \(\mathsf{NoCommit}\), or restart steps is ignored, which is precisely divergence-insensitive weak matching.
\end{proof}

\subsection{Executable-Artifact Evaluation}
\label{app:artifact-bounds}

The performance script runs every sample in a fresh CPython 3.12 process and reports the median of five repetitions on a 16-vCPU AMD EPYC 7R12.
The 19 paper-specific tests exercise indexed pomset resolution, the greatest fixed point, and finite fixtures for all six constructors.
The executable model covers the checker core and all six edit fixtures, while the formal proofs establish immutable edit rules, authority-approved result removal, and the complete rule-update protocol for each policy domain.

\clearpage
\section*{AI Use Acknowledgment}
OpenAI Codex was used throughout the manuscript for initial prose drafting and revision, and in the executable artifact for code scaffolding and test generation.
It also served as an interactive aid for literature discovery, counterexample search, and the discussion and refinement of the formal theory, including definitions, assumptions, theorem statements, and proof structure.
The authors chose and directed the research, made all final theoretical and methodological decisions, independently verified the cited sources, theorem statements, proofs, code, and experimental results, and remain responsible for every claim, proof, result, and citation.

\bibliographystyle{IEEEtran}
\IEEEtriggeratref{14}
\bibliography{references}

\end{document}